\documentclass[oneside,11pt,a4paper]{article}%
\usepackage{amsfonts,amsmath,amssymb,amsthm}
\usepackage[authoryear, round]{natbib}
\usepackage{graphicx} 
\usepackage{xcolor} 
\usepackage{tikz}
\usepackage{caption}
\usepackage{subcaption}
\usepackage{makecell}

\newtheorem{Assumption}{Assumption}
\newtheorem{Proposition}{Proposition}
\newtheorem{Lemma}{Lemma}

\newtheorem{Definition}{Definition}

\newtheorem{Remark}{Remark}
\newtheorem{Example}{Example}

\newcommand{\CE}{\textnormal{CE}}
\renewcommand{\S}{\mathcal{S}}
\title{Stochastic Monotonicity and Random Utility Models: The Good and The Ugly\footnote{Tilburg University, Department of Econometrics and OR, PO box 90153, NL-5000 LE Tilburg, The Netherlands, email: \{h.r.f.keffert, n.f.f.schweizer\}@uvt.nl. We thank Jaap Abbring, Sebastian Olschewski, Bettina Siflinger and Arthur van Soest for very helpful comments and discussions.}\\$ $\\}
\author{Henk Keffert and Nikolaus Schweizer}

\date{September 2024}

\begin{document}

\maketitle

\begin{abstract}
When it comes to structural estimation of risk preferences from data on choices, random utility models have long been one of the standard research tools in economics.   A recent literature has challenged these models, pointing out some concerning monotonicity and, thus, identification problems. In this paper, we take a second look and point out that some of the  criticism -- while extremely valid -- may have gone too far, demanding monotonicity of choice probabilities in decisions where it is not so clear whether it should be imposed.  We introduce a new class of random utility models based on carefully constructed generalized risk premia which always satisfy our relaxed monotonicity criteria. Moreover, we show that \textit{some} of the models used in applied research like the certainty-equivalent-based random utility model for CARA utility actually lie in this class of monotonic stochastic choice models. We conclude that not all random utility models are bad.\par \bigskip 
\noindent \textbf{JEL codes:} C25, D81 \par \bigskip
\noindent \textbf{Keywords:} risk preferences, preference elicitation, random utility model, stochastic choice

\end{abstract}

\newpage
\section{Introduction}
Random utility models (RUMs) are a key component of the applied econometrician's toolkit when the task is to elicit preferences from observed behavior.\footnote{See the Nobel lecture \citet{mcfadden2001economic} for the big picture and \citet{train2009} for a textbook treatment.} The basic idea is that an agent's choices depend on both the agent's preferences for the differences options and on stochastic taste shocks that distort those preferences. The probability that the agent chooses a given option is then just the probability that the distorted utility is larger for this option than for the others. In this paper, we are specifically interested in RUMs as a tool for eliciting risk preferences. By assessing an agent's probability of choosing a safer over a riskier option, one tries to get a sense of how risk-averse the agent is. In this context, RUMs have recently been criticized for exhibiting violations of stochastic monotonicity, e.g. by \citet{Wilcox} and \citet{AB}. The concern here is that in some natural examples the probability of choosing the safer option is not monotonically increasing in the degree of risk aversion so that trying to identify and estimate risk preferences based on these models could be problematic.\footnote{See \citet{wilcox2021utility} for a recent survey of this literature, \citet{bellemare2023estimation} for a textbook treatment of some of its main lessons up to now and \citet{barseghyan2018estimating} for a recent survey of the broader picture. } 

In this paper, we argue that even though this criticism is extremely valid, it may have gone too far in the sense that the monotonicity demands placed on stochastic choice models by this literature may have been too ambitious. For illustration, consider how \citet{Wilcox} motivates his notion of stochastic monotonicity: From \citet{pratt1964}, we know  that an agent with a larger coefficient of absolute risk aversion is more risk averse. From \citet{RS}, we know  that all risk-averse agents prefer the less risky option in any mean-preserving spread (MPS). Thus, in a well-behaved stochastic choice model  the probability of choosing the less risky option in an MPS should increase in the degree of absolute risk aversion. This argument relies on two implicit assumptions: First, it assumes that \citet{pratt1964}'s results about choice between a safe and a risky option extend in the right way to choice between two risky options. Second, it assumes that the strength of the preference for the safer of the  two options increases with the degree of risk aversion. In a similar way, \citet{AB} postulate that choice probabilities should be increasing in the degree of risk aversion for any pair of lotteries for which all agents above a given threshold of risk aversion pick the safer option while all agents below the threshold pick the riskier one. Effectively, this assumes that the problem is so good-natured that a single-crossing property at the level of preferences suffices to guarantee a monotonicity property at the level of choice probabilities. Again, this implicit assumption relies on the intuition that the strength of the preference for the safer option should be increasing in the degree of risk aversion. 

These implicit assumptions are less innocent than they may seem. As was first shown in the literature on choice under background risk \citep{kihlstrom1981,ross1981}, extending the insights of \citet{pratt1964} to choice between two risky lotteries is less straightforward than one may hope. In particular, this literature showed that there exist MPS for which the strength of the preference for the safer option is not monotonic in the coefficient of risk aversion of the underlying utility function. To quantify the strength of the preference, these papers consider risk premia $\pi$ which are the monetary amounts that need to be subtracted from the less risky option to make the agent indifferent between the two options. They showed that these risk premia exhibit non-monotonic behavior in the degree of risk aversion for some MPS. 

For the intuition, start with some risky lottery $Y$. Add a little bit of mean-zero risk to the most favorable outcomes in $Y$ while keeping the least favorable outcome as it is. Call the resulting lottery $X$.\footnote{For an illustration, see Figure \ref{fig3}. The formal definition of ``a little bit of mean-zero risk'' is that after adding it the least-favorable outcomes in $X$ and $Y$ are still identical.} The pair $(X,Y)$ is an MPS. Risk-loving agents will prefer $X$ while risk-neutral agents are indifferent and risk-averse agents prefer $Y$. However, consider an agent who is extremely risk-averse and thus evaluates gambles almost exclusively by their least-favorable outcome. This agent will be close to indifferent between $X$ and $Y$ because the  least-favorable outcomes are the same. For MPS like this one, the strongest preference for the safer gamble $Y$ occurs for risk-averse agents with moderate levels of risk aversion who are still risk-tolerant enough to care about the details of the upside of a gamble. 

By the monotonicity criteria of \citet{Wilcox} or \citet{AB}, a stochastic choice model which does not exhibit monotonic choice probabilities for \textit{all} MPS should be avoided.  This includes many RUMs used in applied work such as those that add stochastic shocks to expected utilities or certainty equivalents of lotteries. In our view, this is too harsh, given that the strength of the risk preference as measured by the premia $\pi$ is not monotonic in risk aversion in general. In examples like the one we just discussed, the problem lies in the pair $(X,Y)$ which is simply not a great pair of lotteries for assessing the degree of risk aversion -- rather than in the stochastic choice model. In particular, any decision criterion that postulates monotonicity of choice probabilities in examples where the strength of the preference vanishes as risk aversion grows only leaves room for ordinal models that look at the sign of the preference but not its strength. In this paper, we show that by suitably relaxing the monotonicity conditions, giving up -- among others -- the monotonicity requirement for \textit{all} MPS, we create room for cardinal stochastic choice models that satisfy stochastic monotonicity while still taking into account the strength of the preference and not just its sign. 

We thus propose a new, weaker notion of stochastic monotonicity. We say that a stochastic choice model is $\Pi$-monotonic if choice probabilities are monotonic in the degree of risk aversion for all $\Pi$-ordered pairs of lotteries. Here, a pair of lotteries is $\Pi$-ordered if a suitably defined, generalized risk premium $\pi$ is monotonic in the degree of risk aversion.\footnote{This definition is parallel to \citet{AB}'s notion of stochastic monotonicity which postulates monotonicity for all $\Omega$-ordered pairs of lotteries. These are lotteries for which $\pi$ switches signs at most once, replacing our monotonicity condition by a weaker single-crossing condition. In particular, every $\Pi$-ordered pair is $\Omega$-ordered and $\Omega$-monotonicity implies $\Pi$-monotonicity.} We then show that many of the gambles used in empirical work are actually $\Pi$-ordered, including choice between a safe and a risky option, as well as  -- at least under CRRA or CARA utility -- choice between binary lotteries with identical probabilities. The set of $\Pi$-ordered gambles is thus rich enough to include, e.g., the gambles used in the single choice list method of \citet{eckel2002sex} and in the multiple price list method of \citet{holt2002risk}.\footnote{We refer to \citet{charness2013experimental} and \citet{holzmeister2021risk} for an overview, discussion and comparison of those and other methods for eliciting risk aversion.} We then introduce a generic $\Pi$-monotonic random utility model which we call the $\Pi$-based RUM. The idea is simply to assign to each choice option as a ``disutility index'' the premium it takes to make this option as good as the highest-utility option. Choice probabilities are then determined by adding shocks to these premia, just like in a standard random utility model. In particular, the $\Pi$-based RUM can be estimated using, e.g., maximum likelihood methods in exactly the same way as a conventional random utility model.

Under CARA utility, the $\Pi$-based RUM coincides with a random utility model that transforms expected utilities into their monetary certainty equivalents before adding stochastic shocks.\footnote{This model has been applied e.g. in \citet{friedman1974risk} and in \citet{von2011heterogeneity}.} Even though it is not $\Omega$-monotonic, we feel that its $\Pi$-monotonicity completely rehabilitates it. However, there is a converse to this observation: We show that, basically, if a RUM that is based on certainty equivalents is $\Pi$-monotonic, the underlying utilities have to be CARA. For example, a RUM that adds stochastic shocks to CRRA certainty equivalents cannot be $\Pi$-monotonic.\footnote{This model has e.g. been studied in the model comparison of \citet{olschewski2022empirical} who attribute it to \citet{stewart2018psychological}.} Indeed, for this model,  we demonstrate fairly drastic violations of monotonicity for some standard binary lotteries used in the experimental literature. Our rehabilitation of the certainty-equivalent-based RUM is thus limited to the CARA case or to choice between a risky and a degenerate, non-risky lottery. In both cases, the difference of certainty equivalents can be interpreted as a risk premium.
 
Let us emphasize our conviction that -- in the broader scheme of things -- our views are very much in line with those of \citet{Wilcox} and \citet{AB}. In particular, we share with them the concerns about monotonicity properties of the standard RUM which simply adds stochastic shocks to expected utilities as exemplified e.g. in Figure 1 of \citet{AB}.\footnote{In this sense, our paper is fairly orthogonal to the previous rehabilitation attempt \citet{conterehabilitating} which focuses mostly on that model.} None of the usual parametrizations of CARA and CRRA utility were designed to give insightful comparative statics of expected utility in the risk aversion parameter. Moreover, unlike certainty equivalents or risk premia, expected utility differences are only identified up to a scaling factor -- so some arbitrariness must remain. Finally, \citet{blavatskyy} has shown that any RUM based on expected utilities must exhibit monotonicity problems even when attention is restricted to choices between risky and degenerate lotteries.\footnote{Throughout the paper, we restrict attention to expected utility as the model that generates the underlying non-stochastic preferences. As discussed, e.g., in \citet{wilcox2021utility}, many of the same issues arise for RUMs that are based on preference models from behavioral economics like prospect theory or rank dependent utility. We refer to \citet{vieider2024decisions} and the references therein for some recent work on behavioral stochastic choice models that model preferences and errors within a single framework, thus ruling out the possibility that the error model introduces inconsistencies which are at odds with the underlying non-stochastic preferences -- because there are no underlying non-stochastic preferences.  While this approach lets many of the problems we discuss in this paper disappear, it does create some new ones. First, the models will tend to be a bit more complex because some flexibility is given up \citep[see][Section 4]{vieider2024decisions}. Second, from the perspective of a policy-maker who wishes to maximize an agent's welfare, it is attractive to have a canonical way of decomposing the agent's behavior into preferences (which are being maximized) and errors (which are being ignored). In RUMs and related models, such a decomposition comes for free.}
 
 \citet{barseghyan2018estimating} have pointed out that the concerns about monotonicity in the standard RUM may be less worrisome if the variance of the shocks is estimated together with the risk aversion parameter because then the variance can potentially absorb some of the monotonicity problems. Conversely, the concerns are especially worrying in empirical applications where the variance is assumed to be the same for a cross-section of individuals while the risk aversion parameters can differ at the individual level.\footnote{The same argument is already found in \citet{von2011heterogeneity} who use it to explain their choice of a RUM based on certainty equivalents rather than one based on expected utilities.} In an empirical illustration based on data from \citet{andersen2008eliciting}, we confirm this intuition: When estimating individual CRRA risk aversion parameters using a standard expected-utility-based RUM for a cross-section of individuals, there is a large discrepancy between estimates obtained from homoskedastic and heteroskedastic model specifications.\footnote{To clarify terminology, we note that papers such as \citet{hey1995experimental} and \citet{buschena2000generalized} have studied heteroskedastic, expected-utility-based RUMs in the sense of allowing the variance to vary between lottery pairs for a given individual. From this perspective, all the models we consider are homoskedastic. However, we do share with \citet{hey1995experimental}'s heteroskedastic approach the basic idea of keeping the underlying non-stochastic expected-utility preferences fixed while varying the way in which errors are introduced.} The heteroskedastic estimates are mostly in line with those obtained from other models like the $\Pi$-based RUM or the random parameter model (RPM) suggested by \citet{AB}. In contrast, the homoskedastic estimates are squashed together on a narrow range, reflecting the fact that any given value of the shock parameter is only adequate for a limited range of risk aversion levels. For the expected-utility-based RUM, we really see two problems coming together here: Utility differences are not just non-monotonic, they also completely change their magnitude as the risk aversion level changes. Consequently, when we study the impact of non-monotonicity on a RUM based on certainty equivalents, we observe a more subtle effect. While the non-monotonicity introduces multimodality in the likelihood, the agreement between estimates from homoskedastic and heteroskedastic model specifications remains good. The reason is that even though certainty equivalent differences are non-monotonic, they still vary over the same monetary scale.

The paper is organized as follows. Section \ref{sec2} introduces the setting, the families of utility functions we study and our generalized notion of risk premium. Section \ref{sec3} introduces $\Pi$-monotonicity, our stochastic monotonicity concept. Section  \ref{sec4} then shows that $\Pi$-monotonicity applies to many of the lotteries studied in experimental work, introduces the $\Pi$-based random utility model and shows that it is $\Pi$-monotonic. Section \ref{Dis} adds further discussion of RUMs based on certainty equivalents, the $\Pi$-based RUM and differences between ordinal and cardinal stochastic choice models. In Section \ref{empirics}, we present some empirical illustrations. Section \ref{sec7} concludes.  All proofs are in Appendix \ref{AppA} while Appendix \ref{AEmpirics} fills in additional details about the empirical illustration.

\section{The setting}\label{sec2}

Consider a family $(u_\theta)$ of strictly increasing\footnote{Throughout, we call a function $f:\mathbb{R}\rightarrow \mathbb{R}$ strictly increasing if $x<y$ implies $f(x)<f(y)$ and increasing if  $x<y$ implies $f(x)\leq f(y)$.  In the same way, we define (strictly) decreasing, convex and concave.} and twice continuously differentiable utility functions $u_\theta:\S\rightarrow \mathbb{R}$ parametrized by $\theta\in \Theta\subset \mathbb{R}$. We assume that the support $\S$ is an interval that is unbounded above, $\S=(s_0,\infty)$ for some $s_0 \in [-\infty,\infty)$, and that $\Theta$ is an interval that is unbounded above and large enough to contain $0$. For simplicity, we also assume $\mathbb{N}_{>0}\subseteq \S$, all non-zero integers lie in $\S$.\footnote{This could easily be relaxed if needed. The advantage is that all lotteries with integer outcomes take values in $\S$.} Unless otherwise noted, we restrict attention to finite lotteries $X$ which take values $(x_1,\ldots,x_{n_X})$, $x_i\in \S$, with probabilities $(p_1,\ldots,p_{n_X})$ where $p_i>0$ for all $i$ and $\sum_{i=1}^{n_X} p_i=1$. We write 
\[
E[u_\theta(X)]= \sum_{i=1}^{n_X} p_i u_\theta(x_i) \;\; \text{and} \;\; \CE_\theta(X) = u_\theta^{-1}(E[u_\theta(X)])
\]
for the expected utility and the certainty equivalent of $X$ under $u_\theta$. We are mainly interested in two families of this type, the constant absolute risk aversion (CARA) family with parameter $\theta=\alpha$, support $\S=\mathbb{R}$ and parameter range $\Theta=\mathbb{R}$ which is given by 
\[
u_\alpha(x) =\begin{cases}
-\frac{1}{\alpha} \exp(-\alpha x) & \text{ for } \alpha \neq 0\\
x & \text{ for } \alpha = 0,
\end{cases} 
\]
and the constant relative risk aversion (CRRA) family  with parameter $\theta=\gamma$, support $\S=(0,\infty)$ and parameter range $\Theta=\mathbb{R}$ which is given by 
\[
u_\gamma(x) =\begin{cases}
\frac{x^{1-\gamma}}{1-\gamma} & \text{ for } \gamma \neq 1\\
\log(x) & \text{ for } \gamma = 1.
\end{cases} 
\]
The properties of the family $u_\theta$ that we will need are summarized in the following assumption which is well-known to be satisfied for the CARA and CRRA families:
\begin{Assumption}\label{A1}
For any lottery $X$, we assume that $\CE_\theta(X)$ is continuous in $\theta$ with $\CE_0(X)=E[X]$ and 
\begin{equation}\label{CElim}
\lim_{\theta \rightarrow \infty} \CE_\theta(X) = \min_{i=1,\ldots,n_X} x_i.
\end{equation}
Moreover, for all $\theta_1 < \theta_2$ and all $x\in \S$
\begin{equation}\label{PR}
\mathcal{A}_{\theta_1}(x) < \mathcal{A}_{\theta_2}(x)
\end{equation}
where $\mathcal{A}_\theta(x)=- \frac{u_{\theta}''(x)}{u_{\theta}'(x)}$ is the Arrow-Pratt coefficient of absolute risk aversion.
\end{Assumption}
This assumption will be in force unless otherwise noted. By \eqref{PR}, we thus assume that the utility functions $u_\theta$ are ordered by their coefficients of absolute risk aversion in the sense of Arrow-Pratt. Moreover, $\theta=0$ corresponds to risk-neutral preferences, $\theta < 0$  to risk-loving preferences and $\theta>0$ to risk-averse preferences. 

Equation \eqref{CElim} postulates that in the limit $\theta \rightarrow \infty$ the preferences $\succeq_\theta$ represented by $u_\theta$ converge to the preferences $\succeq_\infty$ of an infinitely risk-averse agent, who evaluates lotteries only by their worst possible outcome, for any pair of lotteries $X$ and $Y$ with respective outcomes $(x_i)$ and $(y_i)$, 
\[
X \succeq_\infty Y \Leftrightarrow \min_i x_i \geq   \min_i y_i.
\]
Since these infinitely risk-averse preferences will play an important role later on, we discuss their properties a bit further.  
Via $\CE_\infty(X) = \min_i x_i$, we can also define a certainty equivalent for the infinitely risk-averse agent.\footnote{However, since the infinitely risk-averse preferences do not satisfy von Neumann and Morgenstern's continuity axiom, they cannot be written as an expected utility.} For this agent, receiving $c$ for certain is exactly as good as receiving any lottery $X$ with worst-possible outcome $c$. The following elementary lemma summarizes in which sense the limiting preferences are in line with those induced by the family $u_\theta$.
\begin{Lemma}\label{CElimit}
Consider a pair of lotteries with $\CE_\theta(Y)>\CE_\theta(X)$ for all finite $\theta$ above some threshold $\theta_0>0$, i.e., for all sufficiently risk-averse agents. Then we either have $\CE_\infty(Y)>\CE_\infty(X)$ or $\CE_\infty(Y)=\min_i y_i = \min_i x_i=\CE_\infty(X)$.
\end{Lemma}
The lemma simply states that the preferences of an infinitely risk-averse  are aligned  with those of an agent with a sufficiently high but finite risk aversion type with a single exception: The infinitely risk-averse agent is indifferent between lotteries that have the same worst-possible outcome.
\begin{Remark}\label{remEU}
The result of Lemma \ref{CElimit} should be seen in contrast to the observation that limits of expected utilities will typically not lead to meaningful limiting preferences. For instance, for any lotteries $X$ and $Y$ with $\min_i x_i>1$ and $\min_i y_i>1$, we have that CARA and CRRA utilities converge to zero, 
\begin{equation} \label{limit0}
0= \lim_{\alpha \rightarrow \infty} E[u_\alpha(X)]=\lim_{\alpha \rightarrow \infty} E[u_\alpha(Y)]= \lim_{\gamma \rightarrow \infty} E[u_\gamma(X)]=\lim_{\gamma \rightarrow \infty} E[u_\gamma(Y)],
\end{equation}
suggesting complete indifference across a vast class of lotteries -- which has no meaningful relation to the sequence of increasingly risk-averse preferences that approaches this limit. Another  way of making the same point is to note that the alternative family of utility functions
\[
\tilde{u}_\theta(x)= A(\theta) u_\theta(x)
\] 
for an arbitrary positive function $A(\theta)$ can lead to different limiting behavior in \eqref{limit0} even though the underlying preferences are the same and Assumption \ref{A1} remains satisfied.\footnote{For illustration, choose $A(\alpha)=\exp(\alpha x_0)$ for some $x_0>1$ in the CARA case. This will change the limit to $-\infty$ for any lottery that has outcomes below $x_0$.} In other words, the limit in \eqref{limit0} is not identified by the underlying preferences. It depends on the specific parametrization of  the family of utility functions. 
\end{Remark}

Finally, we introduce a generalized notion of risk premium that quantifies how much better a decision maker likes one lottery compared to another,  measuring the strength of the decision maker's preference in monetary terms. Traditionally, e.g. in \citet{pratt1964} and \citet{kihlstrom1981},  risk premia have often been defined as  the non-risky monetary amounts that need to be taken away from the less risky lottery to make it as good (or bad) as the more risky one. However, \citet{kimball1990} has shown that for many purposes one can just as well consider the ``compensating risk premium'' which adds an amount to the more risky option to make it equally attractive to the less risky one.\footnote{In particular, \citet{kimball1990} showed that for choice between a risky and a riskless lottery as in \citet{pratt1964}, a ranking in one type of premium implies a ranking in the other so that one can formally add a ranking in compensating premia as a condition to \citet{pratt1964}'s Theorem 1.} Our definition is based on a simple requirement: We want the premium to exist and be well-defined for any pair of lotteries $X$ and $Y$. We  thus define the compensating premium $\pi_\theta(X,Y)$ via the amount that needs to be added to the lower-utility lottery to make it as attractive as the higher-utility lottery. 
\begin{Definition}[Compensating premium]\label{Defpi}
For any $\theta \in \Theta$ and any pair of lotteries $X$ and $Y$ with $E[u_{\theta}(X)] \leq E[u_{\theta}(Y)]$, we define the non-negative real number $\pi_\theta(X,Y)$ as the unique solution to the equation 
\begin{equation}\label{CP}
E[u_{\theta}(X+ \pi_\theta(X,Y))] = E[u_{\theta}(Y)]
\end{equation}
where $X+\pi_\theta(X,Y)$ denotes the lottery with outcomes $x_i+\pi_\theta(X,Y)$. Moreover, for lotteries $X$ and $Y$ with $E[u_{\theta}(X)] > E[u_{\theta}(Y)]$, we define $\pi_\theta(X,Y)=-\pi_\theta(Y,X)$. 
\end{Definition}

In the appendix, we give a short proof that $\pi_\theta(X,Y)$ as defined exists and is unique. When $X$ is a mean-preserving spread of $Y$ and $u_\theta$ is concave ($\theta>0$) so that the agent is risk-averse, one can easily verify that $\pi_\theta(X,Y)$ corresponds to the compensating \textit{risk} premium as defined by \citet{kimball1990}. However, by defining the premium for any pair of lotteries, we necessarily also define it for pairs of lotteries that cannot be ranked in terms of their riskiness. We thus simply call $\pi_\theta(X,Y)$ the compensating premium.

\begin{Remark}
The guiding observation behind our definition of the compensating premium is that $E[u_\theta(Y+c)]$ is well-defined for any $c\geq 0$ while 
$E[u_\theta(Y-c)]$ may not be well-defined when $Y-c$ takes values below the support $\S$ of $u_\theta$. In particular, it is not always possible to shift down a favorable lottery $Y$ to the point where it matches the utility level of a less favorable one.\footnote{\label{ftex}For an example, consider CRRA utility with $\gamma=0.5$ and thus $u_{0.5}(x)=2 \sqrt{x}$. Under lottery $Y$, the agent gets 2 for certain, while under $X$ the agent receives with equal probability 10 or 1. Here, the agent strictly prefers to receive $10-\pi$ and $1-\pi$ with equal probability over receiving 2 for any $\pi$ with $1-\pi \geq 0$, so it is not possible to equalize utilities by shifting $X$ downwards. Another class of examples that is typical of financial applications -- but violates our focus on finite lotteries -- is that of a CRRA utility agent choosing between two lognormal payoffs with full support on $(0,\infty)$. In this case, none of the two payoffs can be shifted downwards without leaving the support.} In contrast, it is always possible to shift the less favorable lottery upwards to the utility level of the more favorable one. The compensating premium is the unique additive premium that only involves upwards shifts of lotteries.
\end{Remark}

Intuitively, $\pi_{\theta}(X,Y)$ quantifies the utility difference between $Y$ and $X$ in monetary terms. The larger the value of $\pi_{\theta}(X,Y)$, the better is $Y$ compared to $X$. This is easiest to see in the special case of CARA utility in which we have a closed-form expression for the compensating premium:
\begin{Lemma}\label{CARApi}
Under CARA utility, the compensating premium is given by the difference of certainty equivalents
\[
\pi_\alpha(X,Y) = \CE_\alpha(Y)-\CE_\alpha(X).
\]
\end{Lemma}

The lemma follows immediately from the translation invariance condition $\CE_\alpha(X+c)=\CE_\alpha(X)+c$ which holds for CARA utility. The following lemma summarizes some further basic properties of $\pi_{\theta}(X,Y)$.

\begin{Lemma}\label{lempi}
For any $\theta \in \Theta$ and any pair of lotteries $X$ and $Y$, $\pi_{\theta}(X,Y)\geq 0$ is equivalent to $E[u_{\theta}(Y)] \geq E[u_{\theta}(X)]$ and, in particular, $\pi_{\theta}(X,Y)= 0$ corresponds to indifference, $E[u_{\theta}(Y)] = E[u_{\theta}(X)]$.  Moreover, $\pi_{\theta}(X,Y)$ is continuous in $\theta$ and satisfies $\pi_{0}(X,Y)=E[Y]-E[X]$ and 
\[
\lim_{\theta \rightarrow \infty} \pi_{\theta}(X,Y)= \pi_{\infty}(X,Y) = \min_j y_j - \min_i x_i.
\]
where $\pi_\infty(X,Y)$, the compensating premium for the limiting infinitely risk-averse agent, is defined via $
\CE_\infty(X+\pi_{\infty}(X,Y))=\CE_\infty(Y).$
\end{Lemma}

\section{Monotonicity of stochastic choice models}\label{sec3}

 A stochastic choice model $\rho$ is a model that assigns to each alternative in a collection of gambles $X_1,\ldots,X_n$ and each value of the parameter $\theta \in \Theta$ a choice probability $\rho_\theta(X_i)$ which is interpreted as the probability with which an agent with preferences given by $u_\theta$ chooses alternative $X_i$. 
To fix ideas, consider the logit random utility model which has choice probabilities
\begin{equation}\label{logitRUM}
\rho_\theta^{\text{EU-RUM}}(X_i) = \frac{\exp(\lambda E[u_\theta(X_i)])}{\sum_{j=1}^n \exp(\lambda E[u_\theta(X_j)])}
\end{equation}
where $\lambda$ is a precision parameter that captures the noisiness of choices with $\lambda=0$ corresponding to random choice and $\lambda \rightarrow \infty$ to a uniform distribution on the choices that maximize expected utility. Various other examples of stochastic choice models will follow  below. 

The goal now is to formalize the idea that with a good stochastic choice model, more risk-averse agents -- so those with higher values of $\theta$ -- should have a higher probability of choosing less risky options. For an example of the opposite type of behavior,  inspecting \eqref{logitRUM} and recalling Remark \ref{remEU}, one can easily see that 
$$
\lim_{\theta\rightarrow \infty} \rho_\theta^{\text{EU-RUM}}(X_i) = \frac{1}n
$$
for any fixed $\lambda$ if the family of utility functions is the CARA or CRRA family and all lotteries $X_i$ only take values greater than 1.
Intuitively, since all expected utilities converge to zero, choices become increasingly random -- even though a possible preference for less risky choices should actually be becoming stronger.  This inconsistency is one of the main arguments against standard random utility models like \eqref{logitRUM} brought forward by, among others, \citet{Wilcox} and \citet{AB}. 

We now introduce $\Omega$-monotonicity, the monotonicity criterion for stochastic choice models introduced by \citet{AB}, as well as our own weaker  criterion of $\Pi$-monotonicity. Both criteria focus on choice within pairs of lotteries. The criteria are based on notions of $\Omega$-ordered and $\Pi$-ordered pairs of lotteries $(X,Y)$ which correspond to different sets of pairs of lotteries for which one demands monotonic choice probabilities. Throughout, $X$ should be thought of as the more risky option which is preferred by the less risk-averse types of agents while more risk-averse agents prefer $Y$. 

\begin{Definition}\label{def2}
\item[(i)] A pair of lotteries $(X,Y)$ is $\Omega$-ordered if the premium $\pi_\theta(X,Y)$ satisfies the following single-crossing condition: As a function of $\theta$, $\pi_\theta(X,Y)$ switches signs at most once, and if it does then from negative to positive.
\item[(ii)] A pair of lotteries $(X,Y)$ is $\Pi$-ordered if the premium $\pi_\theta(X,Y)$ is increasing in $\theta$.
\item[(iii)] A stochastic choice model $\rho$ is monotonic at the pair of lotteries $(X,Y)$ if $\rho_\theta(X)$ is decreasing in $\theta$.
\item[(iv)] A stochastic choice model $\rho$ is $\Omega$-monotonic if it is monotonic at all $\Omega$-ordered pairs of lotteries.
\item[(v)] A stochastic choice model $\rho$ is $\Pi$-monotonic if it is monotonic at all $\Pi$-ordered pairs of lotteries.
\end{Definition}

$\Pi$-orderedness imposes a monotonicity condition where $\Omega$-orderedness merely imposes a  single-crossing condition. Accordingly, all $\Pi$-ordered pairs are also $\Omega$-ordered and, in particular, every $\Omega$-monotonic stochastic choice model is also $\Pi$-monotonic. We will now argue that in demanding $\Omega$-monotonicity rather than $\Pi$-monotonicity, \citet{AB} are possibly asking too much. As a first step, note that in requiring monotonicity of choice probabilities for every $\Omega$-ordered pair of lotteries,  one basically requires that a single-crossing property at the level of the premia $\pi_\theta$ implies a monotonicity property at the level of choice probabilities.\footnote{\citet{AB} actually formulate their single-crossing property at the level of the utility differences $E[u_\theta(Y)]-E[u_\theta(X)]$ but this is equivalent as seen in Lemma \ref{lempi} above.} Intuitively, the latter is a much stronger type of condition than the former. Just because a function only switches signs once it need not be monotonic. The following lemma illustrates this intuition, showing that sometimes under an $\Omega$-monotonic stochastic choice model the choice probabilities become more extreme while at the same time the strength of preferences as measured by the premium $\pi_\theta$ goes down. 

\begin{Lemma}\label{OmegaPi}
Suppose $\rho$ is $\Omega$-monotonic. Consider an $\Omega$-ordered pair $(X,Y)$ which is not $\Pi$-ordered. Then there exist $\theta_1<\theta_2$ such that $\rho_{\theta_1}(Y) \leq  \rho_{\theta_2}(Y)$ but also $\pi_{\theta_1}(X,Y)>\pi_{\theta_2}(X,Y)$. 
\end{Lemma}

Thus, under $\Omega$-monotonicity, one may find situations in which the choice probability for an alternative $Y$ will go up with $\theta$ while at the same time the strength of the preference for choosing $Y$ rather than $X$ -- as measured by $\pi_\theta(X,Y)$ -- goes down. Let us emphasize that we do not think that $\Omega$-monotonicity of a stochastic choice model is itself a problem. We merely claim that, in light of the lemma,  $\Pi$-monotonicity may be enough to ask: Lack of monotonicity for $\Omega$-ordered pairs that are not $\Pi$-ordered should not disqualify a stochastic choice model because it is not clear that one should expect or impose monotonicity. 

\begin{figure}
\hspace{1.5cm}
\begin{tikzpicture}[->, >=stealth, line width=1.5pt, node distance=2cm, auto]
    \begin{scope}
    \node (A) {$X$};

    \node[right of=A, xshift=2.5cm, yshift=1.5cm] (B) {};
    \node[right of=A, xshift=2.5cm] (C) {};
    \node[right of=A, xshift=2.5cm, yshift=-1.5cm] (D) {};

    \draw (A) -- (B) node[midway, above] {1/3} node[right] {12};
    \draw (A) -- (C) node[midway, above] {1/3} node[right] {9};
    \draw (A) -- (D) node[midway, below] {1/3} node[right] {4};
		\end{scope}
    \begin{scope}[xshift=6.5cm] 
    \node (A) {$Y$};

    \node[right of=A, xshift=2.5cm, yshift=1.5cm] (B) {};
    \node[right of=A, xshift=2.5cm, yshift=-1.5cm] (D) {};

    \draw (A) -- (B) node[midway, above] {2/3} node[right] {10};
    \draw (A) -- (D) node[midway, below] {1/3} node[right] {4};
		\end{scope}
\end{tikzpicture}
\caption{An example of two gambles that are $\Omega$-ordered but not $\Pi$-ordered.}
\label{fig1}
\end{figure}
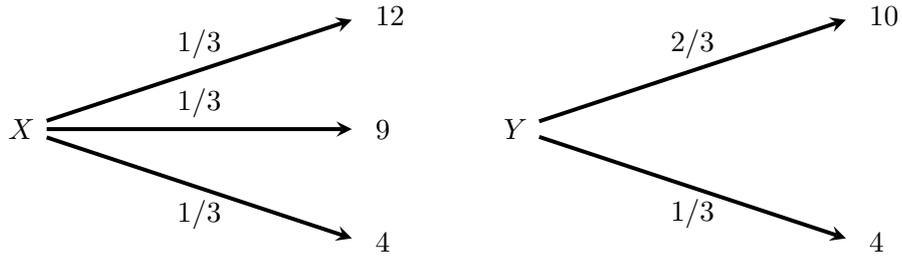

To make this discussion more concrete, consider the pair of gambles in Figure \ref{fig1} which is $\Omega$-ordered but not $\Pi$-ordered. 

\begin{Lemma}\label{lem4912}
The pair of lotteries $(X,Y)$ in Figure \ref{fig1} is $\Omega$-ordered but not $\Pi$-ordered.
\end{Lemma}

The intuition is straightforward: Two distinct types of agents are indifferent between the two lotteries. First, there are agents with intermediate levels of risk aversion $\theta$ who are also indifferent between receiving 10 for certain or 9 and 12 with equal probability. Under CARA utility, these agents correspond to the threshold $\alpha=0.48$ while under CRRA utility, the threshold is $\gamma = 4.91$. Second, there are the infinitely risk-averse agents with $\theta=\alpha=\gamma=\infty$ who are so focused on the downside that they are indifferent between any two lotteries with a worst-case outcome of 4.\footnote{As argued in Lemma \ref{CElimit} above, the preferences of infinitely risk-averse agents can be regarded as the legitimate limiting cases of the preferences that arise under large finite $\theta$. In particular, as argued in Remark \ref{remEU}, the indifference of infinitely risk-averse agents between these particular lotteries $X$ and $Y$ is qualitatively distinct from the ``indifference'' between all lotteries with payoffs greater than 1 suggested by taking the limit of expected CRRA or CARA utility.} In between those two extremes are agents who have a strict preference for the less risky option $Y$. Thus, the two lotteries are $\Omega$-ordered. Yet, not surprisingly, this preference is strongest at some intermediate level of risk aversion. In Figure \ref{fig2} which plots the premia $\pi_\theta(X,Y)$ for the cases of CARA and CRRA utility, we see that the compensating premium is maximal for, respectively, $\theta^*=\alpha^*=0.65$ and   $\theta^*=\gamma^*=6.00$. Accordingly, the two lotteries are not $\Pi$-ordered. 

\begin{figure}[h!]
     \centering
     \begin{subfigure}[b]{0.45\textwidth}
         \centering
         \includegraphics[width=\textwidth]{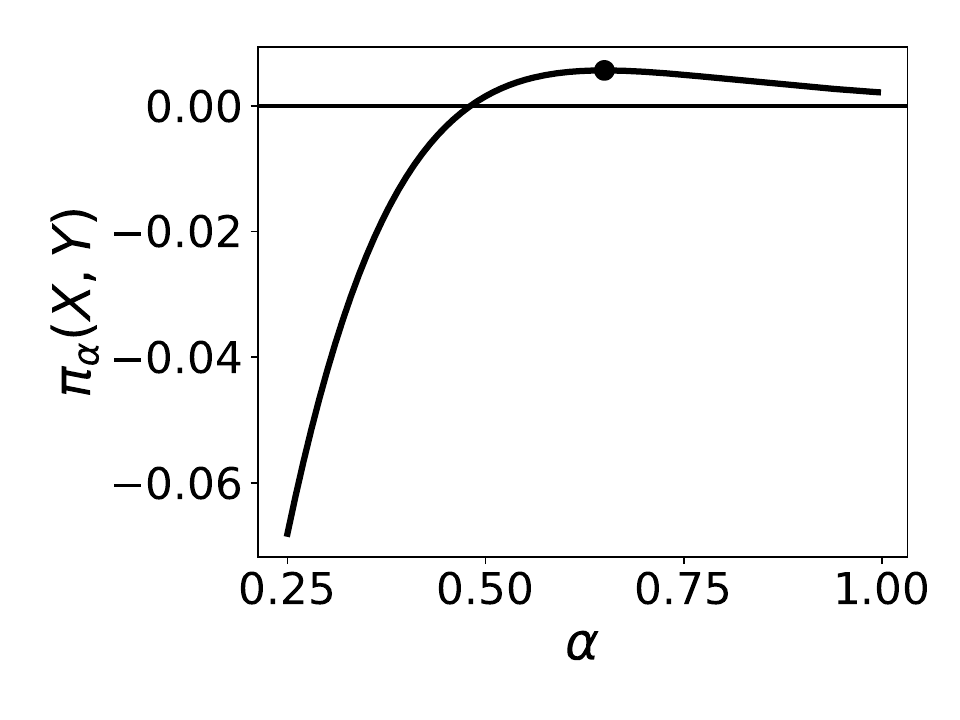}
         \caption{CARA}
         \label{fig:compensating_example_nonmonotonic_CARA}
     \end{subfigure}
     \hfill
     \begin{subfigure}[b]{0.45\textwidth}
         \centering
         \includegraphics[width=\textwidth]{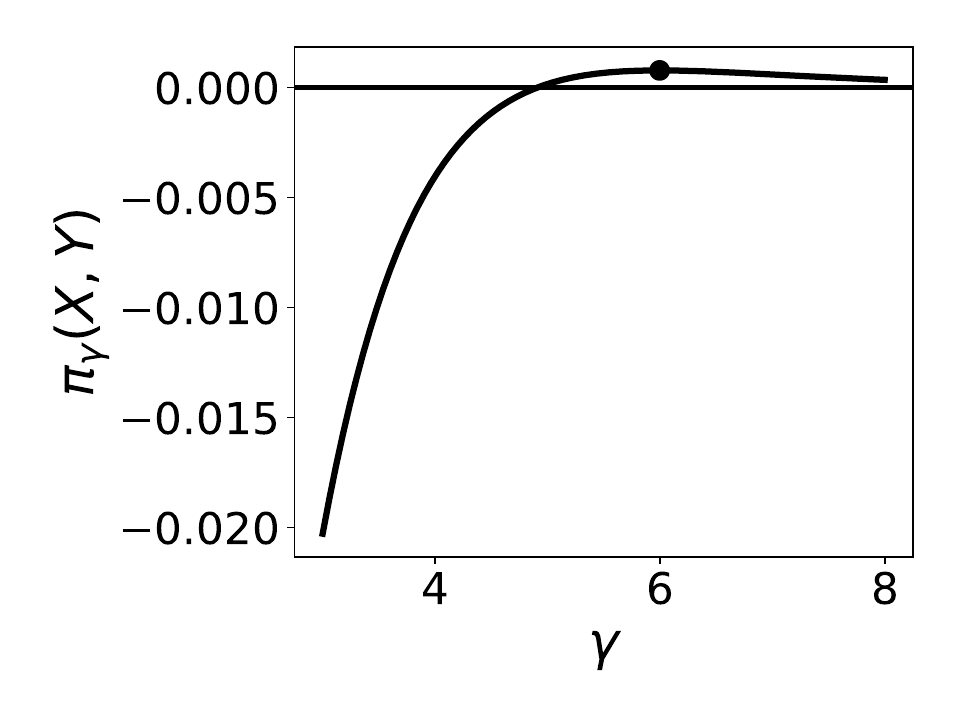}
         \caption{CRRA}
         \label{fig:compensating_example_nonmonotonic_CRRA}
     \end{subfigure}
        \caption{(Non-monotonic) compensating  premium for gambles $X$ and $Y$ from Figure \ref{fig1}. The dots indicate the maximum level of the premium.}
        \label{fig2}
\end{figure}

\begin{Remark}\label{RemPiOmega}
In particular, by Lemma \ref{lem4912}, there always exist $\Omega$-ordered pairs which are not $\Pi$-ordered. It thus follows from Lemma \ref{OmegaPi} that no stochastic choice model can be $\Omega$-monotonic if it has the property that a greater strength of preference in the sense of $\pi_{\theta_1}(X,Y)>\pi_{\theta_2}(X,Y)$ implies a larger choice probability, $\rho_{\theta_1}(Y) >  \rho_{\theta_2}(Y)$. In this sense, $\Omega$-monotonicity implies that the stochastic choice model $\rho$ is an ordinal model.
\end{Remark}

$\Omega$-monotonicity requires that choice probabilities remain monotonic despite the non-monotonicity of the strength of the underlying preferences evident from Figure \ref{fig2}. By suggesting to replace $\Omega$-monotonicity with $\Pi$-monotonicity, we do not suggest that choice probabilities \textit{have to} resemble the pattern of non-monotonicity seen in Figure \ref{fig2}. We merely suggest that violations of monotonicity in an example like this should not be used as evidence against a stochastic choice model. For instance, \citet{AB}'s use an example very much like the one in Figure \ref{fig1} to argue that a random utility model based on certainty equivalents, 
\begin{equation}\label{logitCERUM}
\rho_\theta^{\text{CE-RUM}}(X) = \frac{\exp(\lambda \CE_\theta(X))}{\exp(\lambda \CE_\theta(X))+\exp(\lambda \CE_\theta(Y))},
\end{equation}
should be avoided because it has  choice probabilities converging to 1/2  in the limit $\theta\rightarrow \infty$.\footnote{See the proof of their Corollary 2.} In light of the genuine indifference of the infinitely risk-averse agent, this asymptotic indifference may just as well be a feature and not a bug of the stochastic choice model.\footnote{ Of course, the argument brought forward here is just a refutation of a particular type of counterexample -- not a complete rehabilitation of certainty-equivalent-based random utility models. Indeed, as we discuss below, most certainty-equivalent-based random utility models fail to be $\Pi$-monotonic.}   
\begin{Remark}
\citet{Wilcox} works with a weaker notion of monotonicity than $\Omega$-mo\-no\-to\-ni\-ci\-ty, requiring merely that choice probabilities should be monotonic for all mean-preserving spreads, i.e., for all pairs $X$ and $Y$ where $Y$ can be written as $Y=X+S$ for some $S$ with $E[S|X]=0$. However, in this class, there are still pairs which are $\Omega$-ordered but not $\Pi$-ordered such as the pair of lotteries shown in Figure \ref{fig3}. For this example, both risk-neutral and infinitely risk-averse  agents are indifferent while all agents with $0<\theta <\infty$ strictly prefer $Y$ over $X$, implying once again that the preference is strongest somewhere in between. 
\end{Remark}

\begin{figure}
\hspace{1.5cm}
\begin{tikzpicture}[->, >=stealth, line width=1.5pt, node distance=2cm, auto]
    \begin{scope}
    \node (A) {$X$};

    \node[right of=A, xshift=2.5cm, yshift=1.5cm] (B) {};
    \node[right of=A, xshift=2.5cm] (C) {};
    \node[right of=A, xshift=2.5cm, yshift=-1.5cm] (D) {};

    \draw (A) -- (B) node[midway, above] {1/3} node[right] {11};
    \draw (A) -- (C) node[midway, above] {1/3} node[right] {9};
    \draw (A) -- (D) node[midway, below] {1/3} node[right] {4};
		\end{scope}
    \begin{scope}[xshift=6.5cm] 
    \node (A) {$Y$};

    \node[right of=A, xshift=2.5cm, yshift=1.5cm] (B) {};
    \node[right of=A, xshift=2.5cm, yshift=-1.5cm] (D) {};

    \draw (A) -- (B) node[midway, above] {2/3} node[right] {10};
    \draw (A) -- (D) node[midway, below] {1/3} node[right] {4};
		\end{scope}
\end{tikzpicture}
\caption{A mean preserving spread which is not $\Pi$-ordered.}
\label{fig3}
\end{figure}
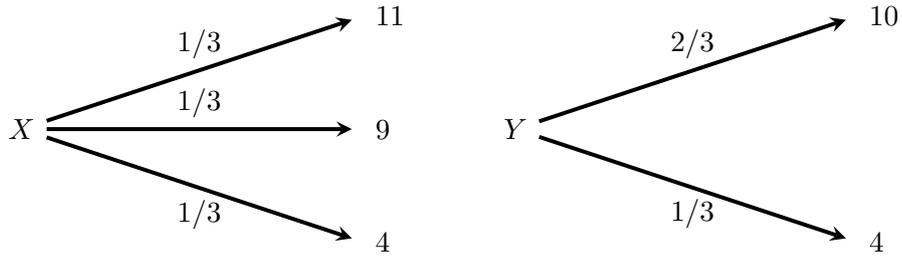

To put the examples from Figures \ref{fig1} and \ref{fig3} into context, it is useful to recall some results from the literature on choice under background risk from \citet{kihlstrom1981} and \citet{ross1981}. These papers study choice between mean-preserving spreads $Y$ and $X=Y+S$, $E[S|Y]=0$ and ask under which conditions a utility function $u$ is more risk-averse than a  utility function $v$ in the sense that the premium for taking up $X$ in place of $Y$ is always greater under $u$ than under $v$. It turns out that a ranking of $u$ and $v$ in terms of their Arrow-Pratt coefficients of risk aversion is not sufficient. \citet{ross1981} derives a necessary and sufficient condition on pairs of utility functions which is  not satisfied for pairs of CARA and CRRA utilities, implying that the parameters $\alpha$ and $\gamma$ do not represent risk aversion for general mean-preserving spreads. There will be mean-preserving spreads for which the more risk-averse agent is closer to being indifferent. The gambles from Figure \ref{fig3} are an example.   \citet{kihlstrom1981} add a positive result: As long as the utility functions $u_\theta$ exhibit decreasing absolute risk aversion and $S$ is independent of $Y$ -- so that the additional risk is the same across all outcomes of $Y$ -- risk premia are monotonic in $\theta$, cf. Lemma \ref{kihlstrom} below. 

\begin{Remark}
Recently, \citet{balter2024} have shown a result that illustrates the intuition that the non-monotonicity of premia for certain MPS arises due to excessive additional risk in favorable outcomes that matter less for more risk-averse types. For gambles $Y$ and $X=Y+hS$ with $h$ sufficiently small, they show that risk premia are increasing with $\theta$ if
\begin{equation}\label{BCS}
\mathcal{A}_\theta(y)\times \textnormal{Var}(S| Y=y)
\end{equation}
is decreasing in $y$. For CARA utility (with $\theta=\alpha>0$) the first factor is constant in $y$, suggesting that the amount of additional risk $\textnormal{Var}(S| Y=y)$ can only be smaller at more favorable outcomes $y$ but not larger. For CRRA utility (with $\theta=\gamma >0$), the first factor scales like $1/y$, suggesting that the variance of the additional risk $\textnormal{Var}(S| Y=y)$ may grow linearly with $y$ but not more strongly than that.
\end{Remark}

\section{From theory to implementation}\label{sec4}
In the previous section, we have proposed monotonicity of choice probabilities for all $\Pi$-ordered gambles as a consistency criterion for stochastic choice models. We thus propose $\Pi$-ordered gambles as a sort of test-set of gambles that can potentially be used for the elicitation of risk preferences when combined with a $\Pi$-monotonic stochastic choice model. In this section, we take two steps towards making this proposal operational. First, we show that most pairs of gambles that researchers have used to elicit risk preferences are actually $\Pi$-ordered. Second, we show that it is straightforward to implement random utility models that are $\Pi$-monotonic.

\subsection{Which pairs of gambles are $\Pi$-ordered?}\label{secPairs}
When discussing their notion of $\Omega$-ordered gambles, \citet{AB} introduce ``three examples of classes of $\Omega$-ordered pairs of gambles often used in applications, which serve to illustrate the large size of the class of $\Omega$-ordered pairs of gambles.'' In the following, we will have a look at each of these three classes and investigate what remains after moving from $\Omega$-ordered pairs to the more restrictive notion of $\Pi$-ordered pairs. The first class of pairs of gambles are pairs where one alternative is not stochastic. These gambles are always $\Pi$-ordered. This is good news because, indeed, these types of pairs frequently used in empirical research, see, e.g. \citet{bruhin2010risk} for one well-known example.

\begin{Lemma}\label{Xy}
Let $(X,Y)$ be a pair of gambles where $Y$ takes the value $y$ with probability 1. Then $(X,Y)$ is $\Pi$-ordered. 
\end{Lemma}

The second class of pairs of gambles are mean-preserving spreads. As discussed in the previous section and in line with classical results from the literature on background risk, we need to exclude those mean-preserving spreads for which risk premia are non-monotonic in the risk aversion parameter $\theta$.  Since the interpretation of an increase in $\theta$ as an increase in the level of risk aversion is problematic anyway when it comes to these pairs of gambles, we feel that not much is lost by excluding them from the set of gambles that can safely be used in the elicitation of risk aversion. Building on the main result of \citet{kihlstrom1981}, we do however have a positive result for  mean-preserving spreads that add the same independent risk $S$ to all outcomes of $Y$ as long as all the utility functions $u_\theta$ satisfy decreasing absolute risk aversion. The latter condition is satisfied, e.g., for CARA utility and for CRRA utility with $\gamma \geq 0$.

\begin{Lemma}\label{kihlstrom}
Let $(X,Y)$ be a pair of gambles with $X=Y+S$ such that $Y$ and $S$ are independent. Moreover, assume that $\mathcal{A}_\theta(x)$ is  decreasing in $x$ for all $\theta$.  Then $(X,Y)$ is $\Pi$-ordered. 
\end{Lemma}

Third, there is the class of pairs of binary lotteries with identical outcomes as used in \citet{andersen2008eliciting}, \citet{von2011heterogeneity} and many other papers. In the following proposition, we verify directly that such pairs of lotteries are $\Pi$-ordered for the classes of CARA and CRRA utility functions. 

\begin{Proposition}\label{binaryProp}
Consider a pair of binary lotteries $(X,Y)$ such that $X$ takes values $a$ with probability $p\in(0,1)$ and $b$ with probability $1-p$ while $Y$ takes values $c$ with probability $p\in(0,1)$ and $d$ with probability $1-p$. Moreover assume that $0<a<b$, $0<c<d$ as well as $d-c<b-a$ so that the outcomes of $Y$ are less spread-out than those of $X$. If the family of utility functions $(u_\theta)$ is the family of CARA or CRRA functions, then the pair $(X,Y)$ is $\Pi$-ordered. 
\end{Proposition}

\begin{Remark}
The fact that the lotteries in Figure \ref{fig1} are not $\Pi$-ordered shows that there is no hope for generalizing this proposition to lotteries with three or more outcomes.
\end{Remark}

Compared to the list of $\Omega$-ordered pairs of lotteries given by \citet{AB}, we thus mainly lose certain mean-preserving spreads. In particular, we see that most of the gambles used in applied work are actually $\Pi$-ordered, at least for the important classes of CARA and CRRA utility functions. Finally, let us emphasize that it is always possible to plot the function $\pi_\theta({X,Y})$ to check whether it is monotonic so that the pair $(X,Y)$ is $\Pi$-ordered. While the results of this section formally prove that the outcome of such a plot will always be positive in important classes of lottery pairs, it is easy to verify $\Pi$-orderedness computationally in any concrete examples of interest.  

\subsection{Which random utility models are $\Pi$-monotonic?}

In this section so far, we have argued that many pairs of lotteries used in applied research satisfy our notion of $\Pi$-orderedness for common choices of families of utility functions. Now, we move on to a next topic that goes hand in hand with the former one: We study whether there are interesting random utility models that are $\Pi$-monotonic so that their choice probabilities have the right monotonicity behavior in the risk aversion parameter $\theta$ when tested on $\Pi$-ordered pairs. 

For a collection of gambles $X_1,\ldots,X_n$, consider a set of preference indices $V_i(\theta)$ that encode how much an agent with utility function $u_\theta$ likes gamble $X_i$. Classical examples include expected utility $V_i(\theta)=E[u_\theta(X_i)]$ and certainty equivalent $V_i(\theta)=\CE_\theta(X_i)$. In a random utility model, one assumes that there is a collection $\xi_1,\ldots,\xi_n$ of independent, identically distributed taste shocks which are drawn from a distribution with a continuous and strictly increasing cumulative distribution function, and that the probability that an agent with preference type $\theta$ chooses gamble $X_i$ is given by 
\begin{equation}\label{genRUM}
\rho^{\textnormal{RUM}}_\theta(X_i)=\mathbb{P}\left(
V_i(\theta)+ \frac{\xi_i}{\lambda} = \max_{j} V_j(\theta)+ \frac{\xi_j}{\lambda}
\right).
\end{equation}
Here, $\lambda>0$ is a precision parameter that adjusts the size of the shocks. This is the general form of a RUM. We write $\mathbb{P}$ to emphasize that  the probability is taken with respect to the shocks $\xi_i$, not the outcomes of the underlying lottery.  Common choices for the distribution of the taste shocks $\xi$ are the normal distribution and the extreme value type I distribution. Combining the latter distribution with $V_i(\theta)=E[u_\theta(X_i)]$ gives the classical $EU$-based logit RUM shown in equation \eqref{logitRUM} while $V_i(\theta)=\CE_\theta(X_i)$ leads to \eqref{logitCERUM}. When choosing between two alternatives $X$ and $Y$, the probability of choosing $Y$ can be written as an increasing function of the difference in the preference indices $V_Y-V_X$ via
\begin{equation}\label{genRUMXY}
\rho^{\textnormal{RUM}}_\theta(Y)=\mathbb{P}\left(
V_Y(\theta)+ \frac{\xi_Y}{\lambda} > V_X(\theta)+ \frac{\xi_X}{\lambda} 
\right) = \Phi(\lambda (V_Y(\theta)-V_X(\theta)))
\end{equation}
where $\Phi$ denotes the strictly increasing cumulative distribution function of $\xi_X-\xi_Y$.

In the remainder of this section, we first introduce the $\Pi$-based RUM, a random utility model that is $\Pi$-monotonic by construction. We then present some positive results about $\Pi$-monotonicity of CE-based RUMs  such as the logit CE-based RUM from equation \eqref{logitCERUM}. For a parameter value $\theta$ and a collection of gambles $X_1,\ldots,X_n$ denote by $X_{\max}(\theta)$ a solution to 
\[
\max_{X\in\{X_1,\ldots,X_n\}} \,E[u_\theta(X)],
\]
i.e., an optimal gamble among the $X_i$ from the perspective of an agent with preference type $\theta$. Then we set 
\[
V_i(\theta)=-\pi_\theta({X_i,X_{\max}(\theta)}).
\]
We call the resulting random utility model $\rho^{\Pi\textnormal{-RUM}}$ the $\Pi$-based RUM. In principle, one should be worried that the definition of $V_i$ is ambiguous if there are multiple optimal gambles among the $X_i$ but this is not the case because the selection among optimal gambles does not matter for the premium:
\begin{Lemma}\label{lemtop}
Consider lotteries $X_1$, $X_2$ and $X_3$ such that $$E[u_\theta(X_1)]\leq E[u_\theta(X_2)]=E[u_\theta(X_3)].$$ Then $\pi_\theta(X_1,X_2)=\pi_\theta(X_1,X_3)$
\end{Lemma}

The intuition behind the $\Pi$-based RUM is simply that the agent compares the differences between the various lotteries in monetary terms. In particular, the agent assigns to each lottery the amount  $\pi_\theta({X_i,X_{\max}(\theta)})$ it would take to  make this lottery optimal. Minus this amount determines the value $V_i$ of the lottery.\footnote{In general, the $\Pi$-based RUM is thus a ``contextual'' stochastic choice model in the sense of \citet{Wilcox}, meaning that $V_i$ depends not only on $X_i$ but on the whole consideration set $X_1,\ldots,X_n$.} An additive error in the calculation of the $V_i$ then determines the choice probabilities via \eqref{genRUM}.

We next verify that the $\Pi$-based RUM is indeed $\Pi$-monotonic:
\begin{Proposition}\label{PiRUMPi}
The $\Pi$-based RUM is  $\Pi$-monotonic.
\end{Proposition}

Since this result may read mildly tautological, let us emphasize two things: First, the $\Pi$-based RUM does not represent the only type of $\Pi$-monotonic stochastic choice model. For example, the RPMs suggested by \citet{AB} are $\Omega$-monotonic which implies that they are $\Pi$-monotonic as well. Second, while the construction of the $\Pi$-based RUM is inspired by our $\Pi$-monotonicity condition, the resulting model is about as easy to estimate as the usual $EU$-based or CE-based RUMs, and with the same techniques. In Section \ref{piRUM} below, we discuss some further properties of the $\Pi$-based RUM. Its estimation is discussed in the context of the empirical illustration in Section \ref{empirics}.

While we are not aware of previous implementation of RUMs based on risk premia, there is one important special case in which our $\Pi$-based RUM coincides with an established model from the literature used, e.g., in \citet{von2011heterogeneity}: For the CARA family of utility functions, the $\Pi$-based RUM coincides with the CE-based RUM, implying in particular that the latter is $\Pi$-monotonic.

\begin{Proposition}\label{CARAmon}
Under CARA utility, the $\Pi$-based RUM leads to the same choice probabilities as the CE-based RUM. Consequently, the CE-based RUM is $\Pi$-monotonic in this case.
\end{Proposition}

There is one other positive result about monotonicity of choice probabilities in the CE-based RUM discussed already, e.g., in \citet{AB} and \citet{bellemare2023estimation}: For pairs $(X,Y)$ where $Y$ always takes the deterministic value $y$, it follows by Theorem 1 of \citet{pratt1964} that 
\begin{equation}\label{DV}
\Delta_V=V_X(\theta)-V_Y(\theta)=\CE_\theta(X)-y 
\end{equation}
is decreasing in $\theta$, implying that the probability of choosing $X$ is decreasing with $\theta$ as well by \eqref{genRUMXY}. Just like in the CARA case, $\Delta_V$ can be interpreted as a risk premium here\footnote{This is however not necessarily the same risk premium as $\pi_\theta({X,Y})$ because it always corresponds to the equivalent risk premium in the sense of \citet{kimball1990}.} since \eqref{DV} can be rewritten into 
\[
u_\theta(y+\Delta_V)=E[u_\theta(X)].
\]

\begin{Remark}\label{rem:blavatskyy}
In line with Definition \ref{def2}, we can call a stochastic choice model $y$-monotonic if it exhibits monotonicity of choice probabilities in $\theta$ for all pairs $(X,Y)$ for which $Y=y$ with probability 1. By Lemma \ref{Xy}, $\Pi$--mo\-no\-to\-ni\-ci\-ty implies $y$-mo\-no\-to\-ni\-ci\-ty, showing that $\Pi$-monotonicity lies in between $\Omega$-monotonicity and $y$-monotonicity. $y$-monotonicity is effectively the notion of stochastic monotonicity used by \citet{blavatskyy} who shows that $EU$-based RUMs are not $y$-monotonic in general. In contrast, the discussion  surrounding \eqref{DV} shows that the CE-based RUM is $y$-monotonic. However, from an applied perspective, $y$-monotonicity may be too weak in many situations since it does not have any implications for choice between pairs of non-degenerate  lotteries.
\end{Remark}
The previous remark together with the positive result in the CARA case may raise the hope that CE-based RUMs have good monotonicity properties more generally, even in cases where the connection to risk premia is less obvious. In the next section, we will see however that this is not the case.

\begin{Remark}\label{RemConEU}
\citet{Wilcox} proposes a variation of the $EU$-based RUM, the ``contextual utility model'', which rescales the utility differences in a way that improves the monotonicity properties of the model for choice within pairs  of lotteries $(X,Y)$ in some cases. The idea is to divide the difference of expected utilities by $u_\theta(z_{\max})-u_\theta(z_{\min})$ where, respectively, $z_{\max}$ and $z_{\min}$ are the largest and smallest outcome that appear in either of the lotteries $X$ and $Y$. \citet[][p.97]{AB} give an example of an MPS for which the model has non-monotonic choice probabilities, showing that the model is not $\Omega$-monotonic under CRRA utility. That counterexample  is not a $\Pi$-ordered pair of lotteries. However, one can verify numerically that the following pair $(X,Y)$ is $\Pi$-ordered but leads to non-monotonic rescaled expected utility differences, showing that the contextual utility model is not $\Pi$-monotonic:  Under $X$, the outcome is 2, 3, 8 or 10 with equal probability, while under $Y$, the outcome is 4 or 7 with equal probability.
\end{Remark}

\section{Discussion}\label{Dis}
\subsection{Limitations of the $\CE$-based RUM}\label{DisCE}
In Proposition \ref{CARAmon}, we saw that the CE-based RUM is $\Pi$-monotonic under CARA utility. Unfortunately, there is a converse to this result: Whenever the CE-based RUM is $\Pi$-monotonic, the underlying family of utility functions must be a family of CARA functions.
\begin{Proposition}\label{onlyCara}
Consider a stochastic choice model $\rho$ for which choice probabilities for pairs $(X,Y)$ of gambles are strictly monotonic functions of the certainty equivalent differences $\CE_\theta(Y)-\CE_\theta(X)$. If $\rho$ is $\Pi$-monotonic then all the utility functions $u_\theta$ for $\theta\in(0,\infty)$ must be CARA functions. 
\end{Proposition}

\begin{figure}
\hspace{1.5cm}
\begin{tikzpicture}[->, >=stealth, line width=1.5pt, node distance=2cm, auto]
    \begin{scope}
    \node (A) {$X$};

    \node[right of=A, xshift=2.5cm, yshift=1.5cm] (B) {};
    \node[right of=A, xshift=2.5cm, yshift=-1.5cm] (D) {};

    \draw (A) -- (B) node[midway, above] {1/2} node[right] {10};
    \draw (A) -- (D) node[midway, below] {1/2} node[right] {8};
		\end{scope}
    \begin{scope}[xshift=6.5cm] 
    \node (A) {$Y$};

    \node[right of=A, xshift=2.5cm, yshift=1.5cm] (B) {};
    \node[right of=A, xshift=2.5cm, yshift=-1.5cm] (D) {};

    \draw (A) -- (B) node[midway, above] {1/2} node[right] {12};
    \draw (A) -- (D) node[midway, below] {1/2} node[right] {10};
		\end{scope}
\end{tikzpicture}
\caption{An example where the premium is constant in $\theta$.}
\label{fig4}
\end{figure}
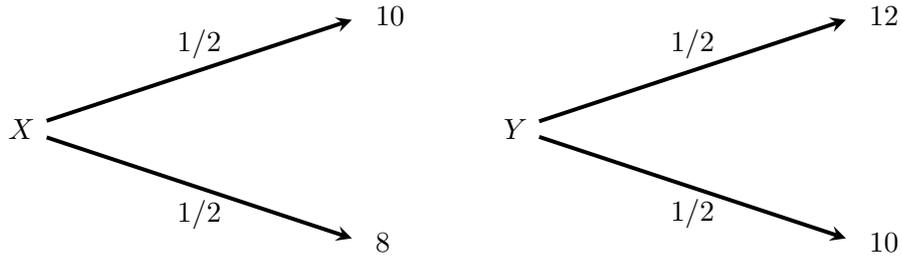

The intuition for this result is based on pairs of gambles like the one in Figure \ref{fig4}. In this example, the difference between the lotteries is actually not stochastic, we have $Y=X+2$ and thus $\pi_\theta({X,Y})=2$. Accordingly, we have a (weakly) $\Pi$-ordered pair of gambles.  For the risk-neutral and the infinitely risk-averse types, $\theta\in \{0,\infty\}$, we also know that the certainty equivalent difference equals, respectively, $11-9=2$ and $10-8=2$. These two types thus have the same choice probabilities. $\Pi$-monotonicity would now require that the certainty equivalent difference is constant in $\theta$ between these two extremes. However, this is generally not the case as illustrated in Figure \ref{fig5} for the case of CRRA utility. 
\begin{figure}[h!]
     \centering
     \includegraphics[width=0.5\textwidth]{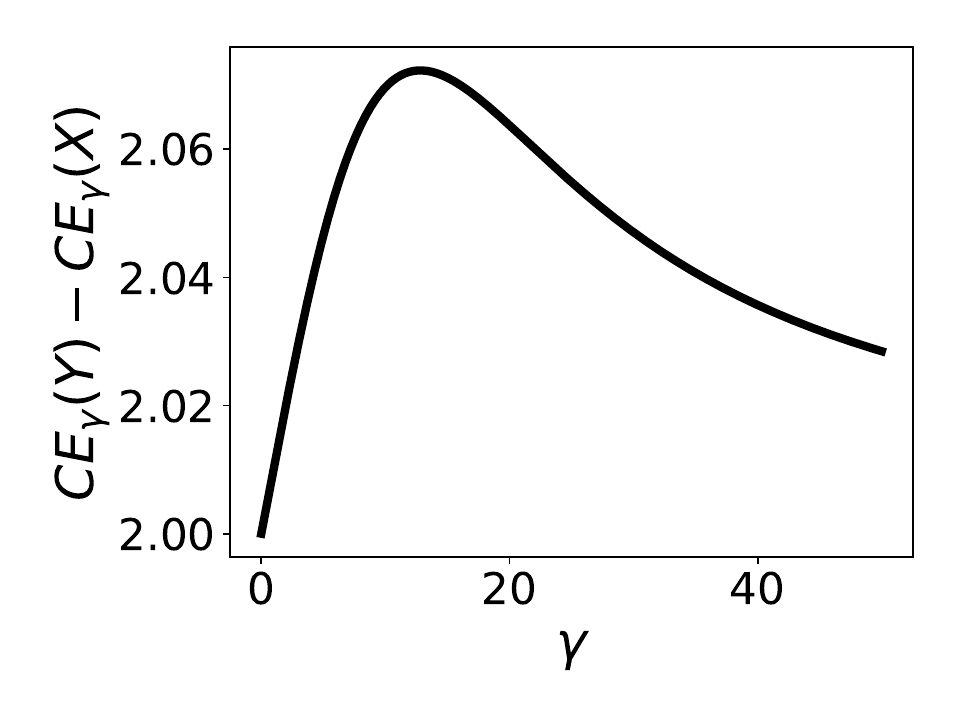}
    \caption{Certainty equivalent difference of gambles $X$ and $Y=X+2$.}
    \label{fig5}
\end{figure}
The proof of the proposition takes this reasoning one step further: $\Pi$-monotonicity for all pairs of the form $(X,X+c)$ can be shown to imply the translation invariance property  $\CE_\theta(X+c)=\CE_\theta(X)+c$. Yet, from the literature, it is known that this property only holds for CARA utility functions.\footnote{In the proposition, we only cover positive values of $\theta$, i.e., risk-averse preferences. The reason is that we have not specified whether the parameter set $\Theta$ is extended into the negative numbers and, if yes, how far. If we add to Assumption \ref{A1} a requirement that preferences converge to those of an infinitely risk-loving agent for $\theta \rightarrow -\infty$, we can easily extend Proposition \ref{onlyCara} to all $\theta \in\mathbb{R}$. This applies, e.g., to CRRA utility.}

Now, clearly, the pair of lotteries in Figure \ref{fig4} is not a typical pair that an applied researcher would use to elicit risk preferences. It is also just a boundary case of $\Pi$-orderedness, corresponding to a constant premium. One might thus hope that a further restriction of the set of admissible test cases -- analogous to the step from $\Omega$-orderedness to $\Pi$-orderedness -- could help to establish monotonicity of the CE-based RUM for general utilities in the examples that actually matter. However, a brief plausibility check shows that such an additional restriction would have to rule out many of the lottery pairs used in empirical research. For example, going through the 40 pairs of lotteries used in the experiments of \citet{andersen2008eliciting} which also appear in our empirical illustration in Section \ref{empirics}, we find that they are all $\Pi$-ordered under CRRA utility by Proposition \ref{binaryProp} but 36 of them exhibit non-monotonicity of choice probabilities for the CE-based RUM,\footnote{The only exceptions are four pairs of choices between degenerate lotteries such as receiving 2000 for certain vs. receiving 4000 for certain which, by themselves, give little insight into risk preferences.} including innocent looking pairs of lotteries such as a choice between 4500 or 50 with equal probability under $X$ vs. 2500 or 1000 with equal probability under $Y$. For the latter pair of gambles, Figure \ref{fig6} displays the certainty equivalent differences, showing that choice probabilities in the CE-based RUM for the safer option $Y$ are \textit{decreasing} in $\gamma$ for $\gamma \geq 1.67$ and thus over an empirically plausible and relevant range of parameter values. In contrast, $\pi_\gamma(X,Y)$ is increasing in $\gamma$, implying that the choice probabilities for the safer option in the $\Pi$-based RUM will increase with risk aversion.
\begin{figure}[h!]
     \centering
     \includegraphics[width=0.5\textwidth]{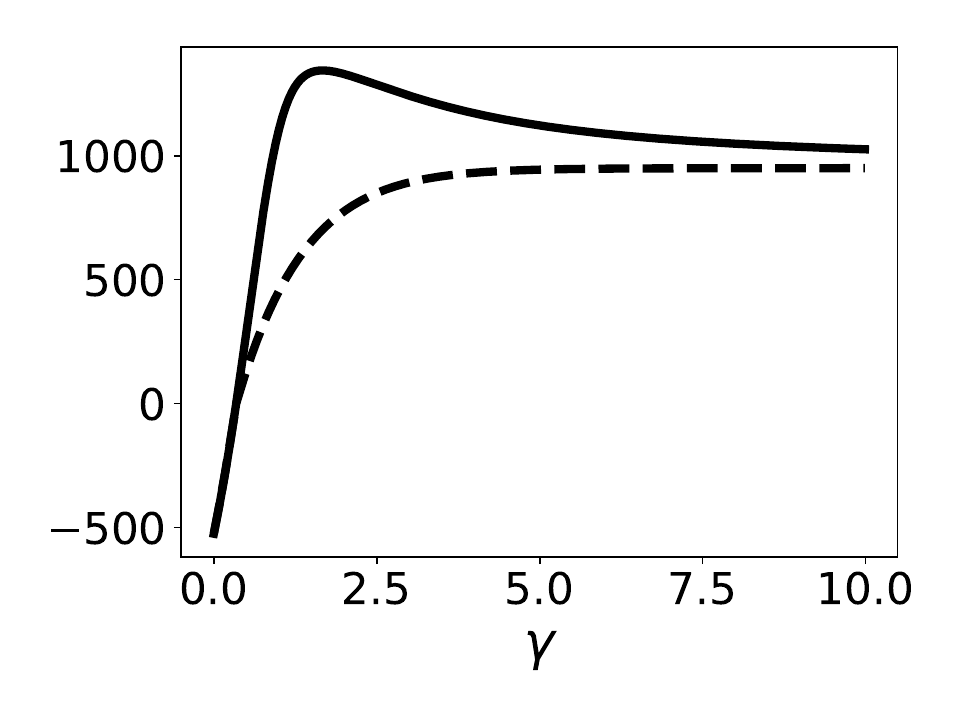}
    \caption{Certainty equivalent difference $\CE_\gamma(Y)-\CE_\gamma(X)$ (solid line) and premium $\pi_\gamma(X,Y)$ (dashed line) for gamble $X$ giving 4500 and 50 with equal probability and gamble $Y$ giving 2500 or 1000 with equal probability for CRRA utility as a function of $\gamma$.}
    \label{fig6}
\end{figure}

\begin{Remark}
One can debate whether the failure of $\Pi$-monotonicity of the CE-based RUM for non-CARA preferences is mainly bad news about the CE-based RUM, or bad news about the use of pairs of binary lotteries for the elicitation of risk preferences. Imagine that you would be asked to consult a group of econometricians from an alternative universe in which the  CE-based RUM is known to be correct: In that universe, it is common knowledge that all agents have CRRA preferences, and that when they evaluate risky lotteries they compare the certainty equivalents and then occasionally make little additive errors in those comparisons.  How should these econometricians estimate the risk preference parameters of individuals? In this artificial and extreme situation, it clearly seems wrong to recommend the use of a miss\-specified model like the $\Pi$-based RUM or the RPM just because of their superior monotonicity properties. Instead,  it seems natural to advise those econometricians against the use of pairs of binary lotteries. Experimental subjects in that universe should rather be asked to choose between a lottery and a safe option.\footnote{Effectively, we would replace our $\Pi$-monotonicity criterion by the $y$-monotonicity from Remark \ref{rem:blavatskyy} where we noted that the CE-based RUM is $y$-monotonic.} However, in our own universe, matters seem less clear.
\end{Remark}

\subsection{Properties of the $\Pi$-based RUM and related models}\label{piRUM}
With the $\Pi$-based RUM $\rho^{\pi\textnormal{-RUM}}$, we have proposed a stochastic choice model that appears to be new beyond the CARA case where it coincides with the CE-based RUM. It may thus be useful to discuss some of its properties and compare them to those of established models. 
 
\begin{Proposition}\label{PiRUMprop}
The choice probabilities $\rho^{\pi\textnormal{-RUM}}_\theta(X_i)$ are continuous in $\theta$. Moreover, in the limit of vanishing noise, the choice probabilities of all dominated options go to zero, $E[u_\theta(X_i)]<E[u_\theta(X_j)]$ for some $j$ implies
\begin{equation}\label{vanishingnoise}
\lim_{\lambda \rightarrow \infty} \rho^{\pi\textnormal{-RUM}}_\theta(X_i) =0,
\end{equation}
and the choice probabilities of optimal gambles are always higher than those of other gambles, $E[u_\theta(X_i)]\geq E[u_\theta(X_j)]$ for all $j$ implies
\begin{equation}\label{topordering}
\rho^{\pi\textnormal{-RUM}}_\theta(X_i) \geq   \rho^{\pi\textnormal{-RUM}}_\theta(X_j)
\end{equation}
for all $j$ and all $\lambda \geq 0$. 
\end{Proposition}

The $\Pi$-based RUM shares the continuity property  with the CE-based and $EU$-based RUMs\footnote{Strictly speaking, this needs the additional assumption that the class of utility function $u_\theta$ is chosen in such a way that continuity of differences of expected utilities is guaranteed.} as well as with RPMs.\footnote{Under an RPM, choice probabilities are given by
\[
\rho^{\textnormal{RPM}}_\theta(X_i)=\mathbb{P}\left(E[u_{\theta+ \xi/\lambda}(X_i)] = \max_j E[u_{\theta+\xi/\lambda}(X_j)]\right),
\]
so a shock $\xi$ is added to the preference parameter and the choice probability of option $i$ is the probability that an agent with preference type $\theta+\xi/\lambda$ considers option $i$ to be optimal. Here, $\lambda$ is a precision parameter and typical distributions for $\xi$ are the Gaussian or the logistic.} The same holds for the property that the probability of choosing dominated options goes to zero as $\lambda$ goes to infinity. The RPM does not have a property like \eqref{topordering} in general. In contrast, the CE-based and $EU$-based RUMs satisfy an even stronger ordering property. For these models, the ranking of choice probabilities reflects the ranking of the expected utilities so, for instance, the third best option also has the third highest probability of being chosen. The reason is simply that in these models the $V_i$ are monotonic transformations of the expected utilities. In contrast, in the  $\Pi$-based RUM, how easy it is to mistake a gamble for the best is not tied as strongly to how good the gamble actually is. The following example illustrates this:
\begin{Example}
Consider three gambles $X_1$, $X_2$ and $X_3$ such that $X_1$ pays 4 for certain, $X_2$ pays with equal probability $1$ or $10$ while $X_3$ pays with equal probability $2$ or $3$. A CRRA utility agent with $\gamma=4$ computes their certainty equivalents as $\CE_4(X_1)=4$, $\CE_4(X_2)=1.82$ and $\CE_4(X_3)=2.4$. Thus, the preference ordering is $X_1 \succ X_3 \succ X_2$.  However if we look at the premia, i.e., the amounts that need to be added to make $X_2$ and $X_3$ equally good as $X_1$, we find $\pi_4({X_2,X_1})=1.42$ and $\pi_4({X_3,X_1})=2.56$. Thus, even though $X_2$ is the third-ranked option, it is easier to improve (or, in our RUM, to misperceive as the best option) than the second-ranked option $X_3$.
\end{Example}

It is also possible to specify a $\Pi$-monotonic RUM with the property that higher utilities always lead to higher choice probabilities. We call this model the cumulative-$\Pi$-based RUM, $\rho^{\textnormal{cum}-\pi\textnormal{-RUM}}$. To this end, denote by $r_\theta(i)$ the rank of $X_i$ among the $n$ gambles $X_1,\ldots,X_n$, i.e. $r_\theta(i)=k$ means that $X_i$ is ranked $k$th when ranking the alternatives by their expected utilities with utility function $u_\theta$ (where ties are broken e.g. lexicographically). Moreover, denote by $X_{(1)}(\theta),\ldots,X_{(n)}(\theta)$ a reordering of the $n$ alternatives according to their expected utilities so that $X_i=X_{r_\theta(i)}(\theta)$. Then we define the preference index $V_i(\theta)$ as $V_i(\theta)=0$ if $r_\theta(i)=1$ and otherwise\footnote{In the title of the paper, we refer to ``the good and the ugly'' with the intention of highlighting that even though a lot has been written recently about the \textit{bad} RUMs there are also the others RUMs which are not bad. Beauty is clearly in the eye of the beholder, and we do not want to want to be judgmental. That being said, the title is partly inspired by the cumulative-$\Pi$-based RUM.}
\[
V_i(\theta)=-\sum_{j=2}^{r_\theta(i)} \pi_\theta\left({ X_{(j)}(\theta),X_{(j-1)}(\theta)}\right).
\] 
Thus, instead of asking by how much $X_i$ needs to be improved to make it optimal, we add up what it takes to move $X_i$ to the next highest rank, to move the next highest ranked option up one more step and so on. For instance, if $X_i$ is ranked third, we add the premium it takes to make it second-ranked to the premium it takes to make the second-ranked option first-ranked. This construction guarantees that the ranking of the $V_i(\theta)$ is aligned with the ranking of the utilities since we only add up non-negative numbers and then put a minus in front. Moreover, the model coincides with the non-cumulative $\Pi$-based RUM 
when it comes to choice between two lotteries. Thus, it inherits $\Pi$-monotonicity which is only based on choice within pairs of lotteries. As a less desirable feature however the cumulative-$\Pi$-based RUM has choice probabilities that are discontinuous in the parameter $\theta$ with more than two choice options.\footnote{Intuitively, the continuity problem arises as follows: Suppose there has been a change of ranking between the second- and third-ranked option at some level $\theta$. The premium between the second and third ranked option is zero here and thus not a problem. However, with a change of the second-ranked option, there can be a jump in the amount it takes to move the second-ranked option to the utility level of the first-ranked one.} By and large, we thus tend to prefer the basic $\Pi$-based RUM even though, of course, the real question is which of the two is closer to how people actually make decisions. 

\subsection{Monotonicity, ordinal models and cardinal models}\label{RPMs}
One could argue that with the RPM, there already is a stochastic choice model that satisfies $\Omega$-monotonicity which implies $\Pi$-monotonicity -- so why look any further. To this end, it is useful to remember two things. First, since different stochastic choice models capture different ways in which people make mistakes, there is little reason to expect that any one of them consistently outperforms the others. Different mistakes may be typical of different situations and require different modeling tools. Second, $\Omega$-monotonicity or $\Pi$-monotonicity by themselves do not make a good model.  For instance, a purely random choice model which assumes that the agent picks between all options with arbitrary but fixed probabilities regardless of $\theta$ is $\Omega$-monotonic and thus $\Pi$-monotonic. Similarly, a random utility model with 
$$V_i(\theta)=1_{\{E[u_\theta(X_i)] = \max_j E[u_\theta(X_j)]\}}$$
is $\Omega$-monotonic. This model is an ordinal version of a RUM in the sense that  choice probabilities only depend on whether a given option is a utility maximizer or not.\footnote{This model can be rewritten into the constant error/tremble model of \citet{harless}, see also \citet{blavatskyy}.} While this model has its merits, it is not particularly flexible, predicting that decision probabilities only change with $\theta$ at indifference points. On an abstract level, having a stochastic choice model that satisfies a given monotonicity concept is only one side of a tradeoff where the other side is to have a model that is flexible and expressive. The best models score well in both dimensions.

In our empirical illustration of Section \ref{empirics}, we find that the RPM and the $\Pi$-based RUM give fairly similar results,  when compared to non-monotonic models like the $EU$-based RUM. In a way, this similarity is encouraging as it suggests that the  elicited risk aversion levels are not overly affected by the stochastic choice model when choosing between pairs of binary lotteries. In more complex settings with more than two outcomes, more than two lotteries or with a multi-parameter preference model, we expect greater differences between stochastic choice models. In particular, we expect that it matters more whether choice probabilities are based only on preference rankings or also on the strength of preferences.\footnote{See \citet{alos2021choice} for some empirical evidence that choice probabilities are affected by how strong preferences are.}

For illustration, consider Figure \ref{fig7}, a variation of Figure \ref{fig1} with a pair of three-outcome lotteries where $\varepsilon\geq 0$ is a parameter that should be thought of as a small number. 
\begin{figure}
\hspace{1.5cm}
\begin{tikzpicture}[->, >=stealth, line width=1.5pt, node distance=2cm, auto]
    \begin{scope}
    \node (A) {$X$};

    \node[right of=A, xshift=2.5cm, yshift=1.5cm] (B) {};
    \node[right of=A, xshift=2.5cm] (C) {};
    \node[right of=A, xshift=2.5cm, yshift=-1.5cm] (D) {};

    \draw (A) -- (B) node[midway, above] {1/3} node[right] {14};
    \draw (A) -- (C) node[midway, above] {1/3} node[right] {6};
    \draw (A) -- (D) node[midway, below] {1/3} node[right] {4};
		\end{scope}
    \begin{scope}[xshift=6.5cm] 
    \node (A) {$Y$};

    \node[right of=A, xshift=2.5cm, yshift=1.5cm] (B) {};
    \node[right of=A, xshift=2.5cm, yshift=-1.5cm] (D) {};

    \draw (A) -- (B) node[midway, above] {2/3} node[right] {9};
    \draw (A) -- (D) node[midway, below] {1/3} node[right] {$4-\varepsilon$};
		\end{scope}
\end{tikzpicture}
\caption{Discontinuity at infinity.}
\label{fig7}
\end{figure}
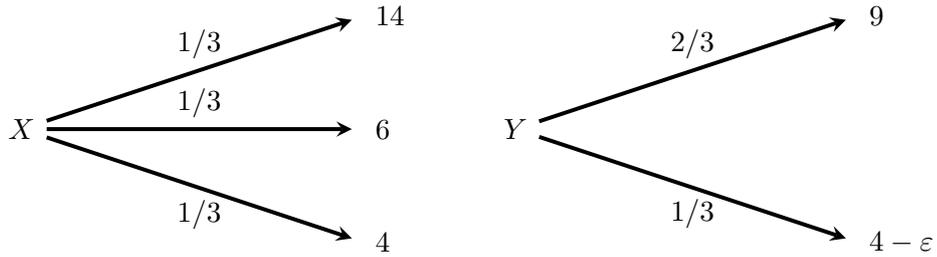
For small $\varepsilon>0$, risk-averse agents with intermediate values of $\theta$ will prefer $Y$ while agents with sufficiently large \textit{or} sufficiently small $\theta$ will have a preference for $X$. Under models like the $\Pi$-based or CE-based RUM, the choice probabilities in the limit $\theta\rightarrow \infty$ will be arbitrarily close to 50/50 for small $\varepsilon$ because $$\lim_{\theta \rightarrow \infty}\pi_\theta(X,Y)=\lim_{\theta \rightarrow \infty}\CE_\theta(Y)-\CE_\theta(X)=-\varepsilon.$$ Ordinal models like the RPM cannot account for what happens in the limiting case $\theta = \infty$ and must thus exhibit a ``discontinuity at infinity''. Under the standard RPM we have 
$\lim_{\theta \rightarrow \infty} \rho^{\textnormal{RPM}}_\theta(X)=1$ for any $\varepsilon >0$ but also $\lim_{\theta \rightarrow \infty} \rho^{\textnormal{RPM}}_\theta(X)=0$ for $\varepsilon =0$, implying that choice probabilities of sufficiently risk-averse should react very sensitively to changes in $\varepsilon$ around 0.\footnote{For example, for $\varepsilon=0.01$ under CRRA utility, agents with $\gamma$ between 1.22 and 10.16 will prefer $Y$ while agents with higher or lower $\gamma$ prefer $X$. Thus, for large $\gamma$ the RPM probability of choosing $X$ goes to 1. In contrast, for $\varepsilon=0$, all types with $\gamma >1.20$ prefer $Y$ so the RPM probability of choosing $X$ goes to 0.} In principle, these observations open the door for empirical research that discriminates between the ordinal and cardinal stochastic choice models. We leave it to future work.

\section{Empirical Illustration}\label{empirics}

In our empirical illustration, we rely on data from a representative sample of the Danish population due to \citet{andersen2008eliciting} that is also used in the empirical part of \citet{AB}. While \citet{AB} focus on a pooled model where one set of parameters is estimated jointly for all of the 253 participants, we consider models that estimate individual risk aversion parameters for all individuals.\footnote{Table \ref{tab:pooled_estimates} in the appendix displays results for a pooled estimation as in \citet{AB}'s Table 1 but extended to the larger set of choice models we consider.} We are particularly interested in investigating the intuition provided by \citet{barseghyan2018estimating} that monotonicity problems have the strongest impact in homoskedastic models that allow for heterogeneity in the risk aversion parameters while imposing homogeneity in the parameters that govern the stochasticity of choices. 

In addition to the $EU$-based RUM and the RPM investigated by  \citet{AB}, we also include in our exercise the $\Pi$-based RUM, the CE-based RUM and the contextual utility model of \citet{Wilcox} which we denote by Con$EU$ in the following.\footnote{Here, the contextual model represents the wider class of models that take the $EU$-based RUM as a starting point and add some sort of normalization or rescaling to mitigate the monotonicity problems. See  \citet{olschewski2022empirical} for some further discussion of different ways of normalizing expected utility; and \citet{AB} for some additional discussion of the potential limitations of such attempts.} 
 For setting up the models and their estimation, we follow  \citet{AB} in most of the choices: We focus on CRRA utility\footnote{This focus is also natural for studying the difference between the CE-based and $\Pi$-based RUMs which would be absent under CARA utility. In principle, it is recommended to account for initial wealth when estimating CRRA utility. We abstract from this complication in this   methodological exercise.} with parameter $\gamma$ and on RUMs with errors following an extreme value type I distribution with precision parameter $\lambda$. For the RPM, we choose a logistic error distribution, leading to a similar logit-type closed-form expression for choice probabilities. In all models, we also include a tremble probability $\kappa\in [0,1]$ which is the probability with which subjects misclick and pick the opposite of what they would have chosen otherwise.\footnote{A tremble probability like this is key to make RPMs operational with real data: Without trembles,  RPMs cannot account for behaviors that are not optimal under some realization of the random parameter $\theta$. Allowing for $\kappa>0$ introduces misclicks as an alternative explanation for behaviors that would otherwise be inconsistent with the RPM. To keep things comparable, we also introduce $\kappa$ in the other models.}  When participants indicate that they are indifferent between $X$ and $Y$, we assign a half choice to each of the alternatives. We estimate all models by maximum likelihood as discussed in detail in Appendix \ref{AEmpirics}. 

Participants in the study of \citet{andersen2008eliciting} faced four blocks of questions choosing between lotteries with fixed payoffs and varying probabilities in the style of \citet{holt2002risk}. In the first block, the choice is between 3850 and 100 with probability $p$ and $1-p$ versus 2000 and 1600 with $p$ and $1-p$ for $p=0.1,0.2,\ldots, 1$. In the other three blocks, the four payoffs are, respectively, 4000 and 500 versus 2250 and 1500; 4000 and 150 versus 2250 and 1500; and 4500 and 50 versus 2500 and 1000. While 116 of the 253 participants were asked all ten questions in all four blocks,  two other groups of 67 and 70 participants only received a selection of, respectively, questions 3, 5, 7, 8, 9, 10 or questions 1, 2, 3, 5, 7, and 10 in each block, leading to a total of 24 questions.

While the last question in each block is a choice between degenerate lotteries such that any agent with an increasing utility function should prefer $X$ over $Y$, the remaining 36 questions have unique indifference levels of $\gamma$. The smallest and largest such indifference levels are given by -1.84 and 2.21, implying that for an agent who deterministically answers all 40 questions in line with CRRA utility, we can tell whether their $\gamma$ is  above 2.21, below -1.84 or in some interval in between as determined by the 36 indifference levels.\footnote{As a small subtlety, those subjects who only received 24 questions have, depending on their group assignment, a smallest indifference level of -0.52 or a largest indifference level of 1.16 while the other boundary remains at, respectively, 2.21 or -1.84.} In particular, it seems fair to conclude that the data set is most suitable for identifying values of $\gamma$ in the range from -1.84 and 2.21. Any estimated values outside of this range are only identified by the functional form of the respective stochastic choice model and should be taken with a grain of salt.

\subsection{Homoskedastic vs. heteroskedastic estimation}

In the five panels of Figure \ref{fig:all_gamma_homo_hetero}, we compare for each of the five stochastic choice models the individual risk aversion estimates $\gamma_i$ in a heteroskedastic and a homoskedastic specification of the model. In the heteroskedastic specification, each of the 253 participants has an individual value of the trembling probability $\kappa_i$ and the precision parameter $\lambda_i$ while the homoskedastic specification assumes that $\kappa$ and $\lambda$ are the same for all individuals.\footnote{As pointed out, e.g., in \citet{von2011heterogeneity}, separate identification of the two error parameters $\kappa_i$ and $\lambda_i$ from small amounts of data may be an issue. However, our focus is on recovering $\gamma_i$. Since we mostly find a good agreement between the heteroskedastic and homoskedastic models, we do not expect much additional insights from a restricted version of the heteroskedastic model ($\kappa$ and $\lambda_i$; or $\lambda$ and $\kappa_i$) which can be thought of as a middle ground between those extremes.} Points that lie close to the diagonal in each of the scatter plots indicate that the estimated $\gamma_i$  from the   heteroskedastic and homoskedastic models lie closely together for individual $i$. For convenience, all values outside the range $[-5,5]$ are projected to $-5$ or $5$. Points with $\kappa_i>0.25$, which are closer to random choice ($\kappa_i=0.5$) than to following the underlying CRRA model are depicted as triangles. 

For the models in panels (b) to (e), we see a good match between the homoskedastic and heteroskedastic estimates with, depending on the model, at least 199 of the 253 estimated $\gamma_i$ within a distance of 0.2 of each other and at least 165 at a distance of less than 0.1. A fair number of the remaining points lie outside the range from -1.84 to 2.21 where we can expect the $\gamma_i$ to be identified accurately, including, e.g. 14 participants (12 risk-averse, 2 risk-loving) whose behavior is consistent with (non-stochastic) CRRA utility for some $\gamma_i$ above the highest or below the lowest indifference-$\gamma$ that appeared in their choice lists. Quite a few of the other outliers -- especially those that do not lie at the $\pm 5$-boundary -- are marked with a triangle, suggesting that here the CRRA model is a poor behavioral fit.\footnote{\label{3people}For illustration, there are three participants who always choose lottery $Y$, thus picking the less risky option for all non-degenerate lottery pairs and the smaller amount of money for all degenerate ones. Expected utility can only rationalize this by combining sufficiently strong risk aversion with a distaste for money, i.e., a decreasing and concave utility function. While our CRRA model only allows for increasing utility, the stochastic choice models can perfectly match this behavior by combining risk-lovingness and taste for money ($\gamma_i \rightarrow -\infty$) with an  agent  who consistently picks the opposite of their true preference ($\kappa_i=1$ and $\lambda_i\rightarrow \infty$). Another three participants claim indifference between all pairs of lotteries such that random choice ($\lambda_i \rightarrow 0$ or $\kappa_i\rightarrow 1$ with arbitrary $\gamma_i$) maximizes their likelihood contribution.}

\begin{figure}[h!]
     \centering
     \begin{subfigure}[b]{0.45\textwidth}
         \centering
         \includegraphics[width=\textwidth]{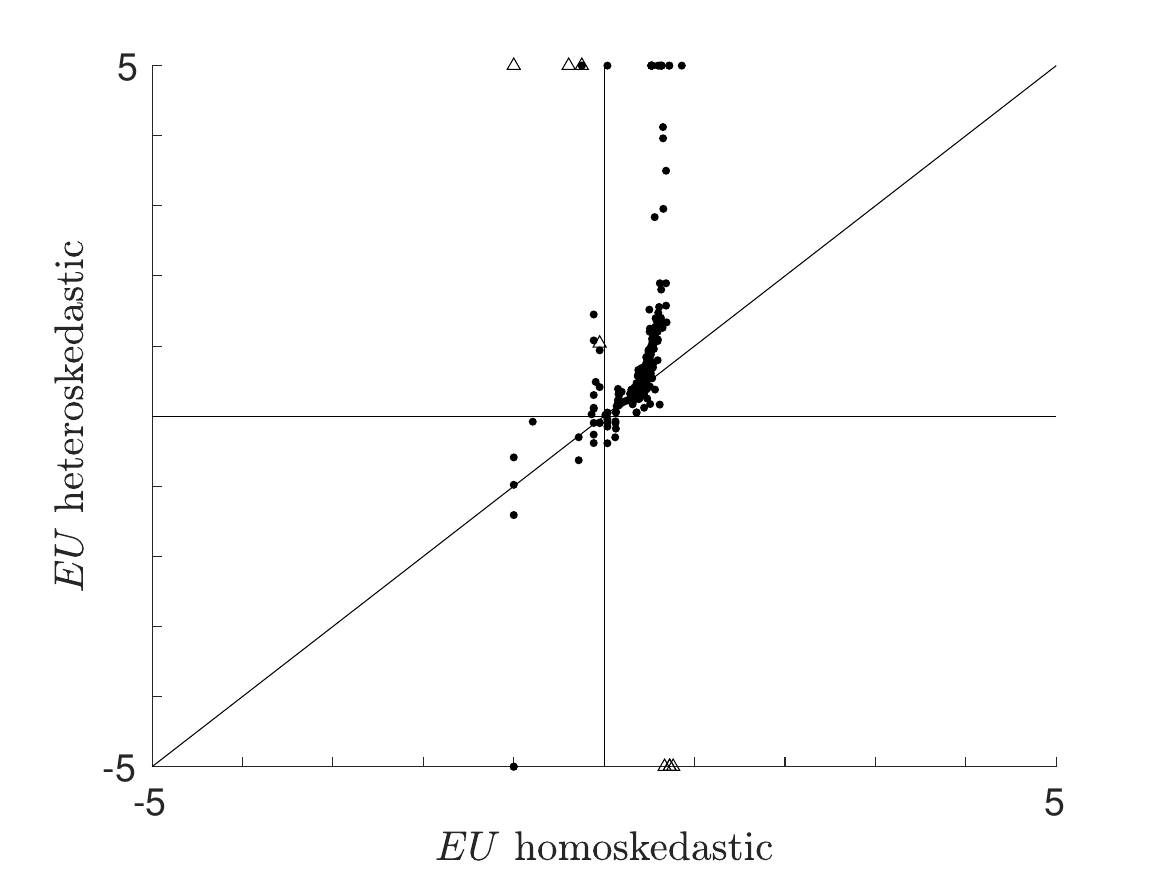}
         \caption{}
         \label{fig:eu_gamma_homo_hetero}
     \end{subfigure}
     \hfill
     \begin{subfigure}[b]{0.45\textwidth}
         \centering
         \includegraphics[width=\textwidth]{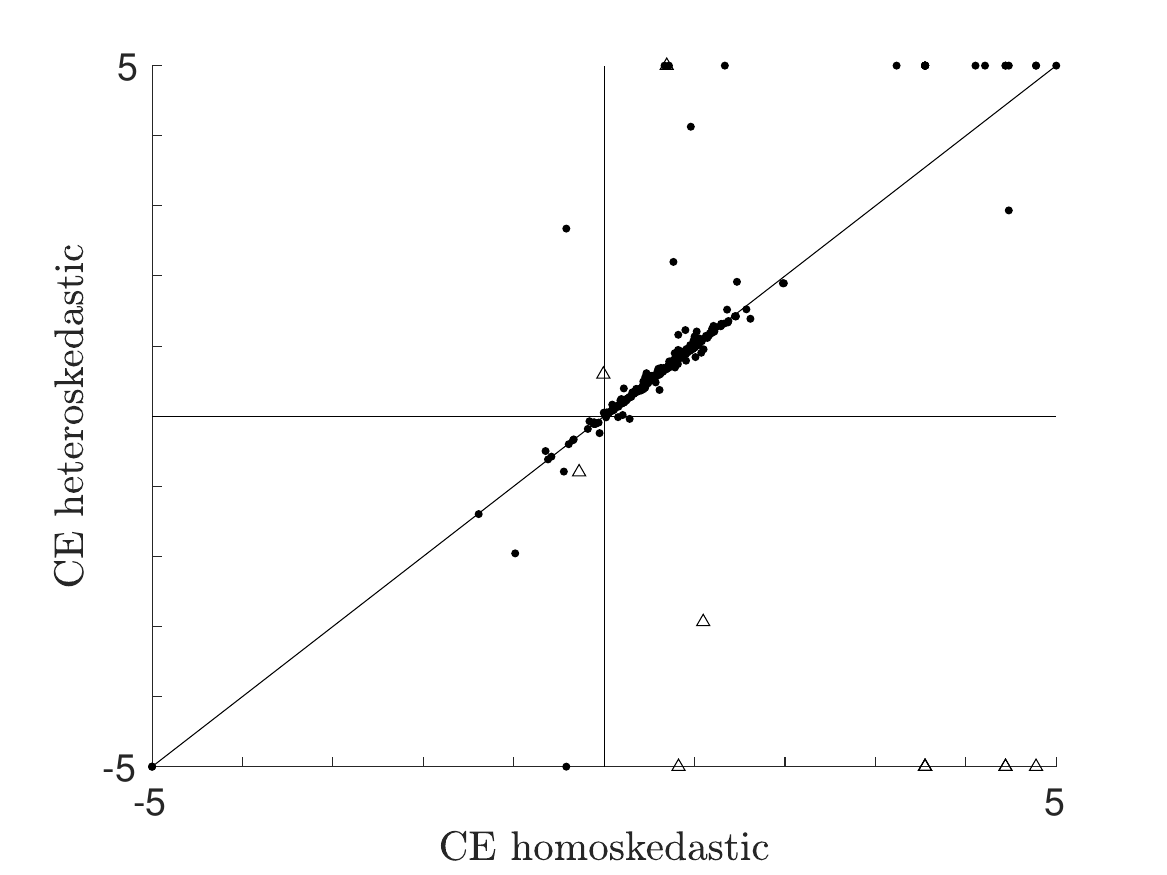}
         \caption{}
         \label{fig:ce_gamma_homo_hetero}
     \end{subfigure}
      \begin{subfigure}[b]{0.45\textwidth}
         \centering
         \includegraphics[width=\textwidth]{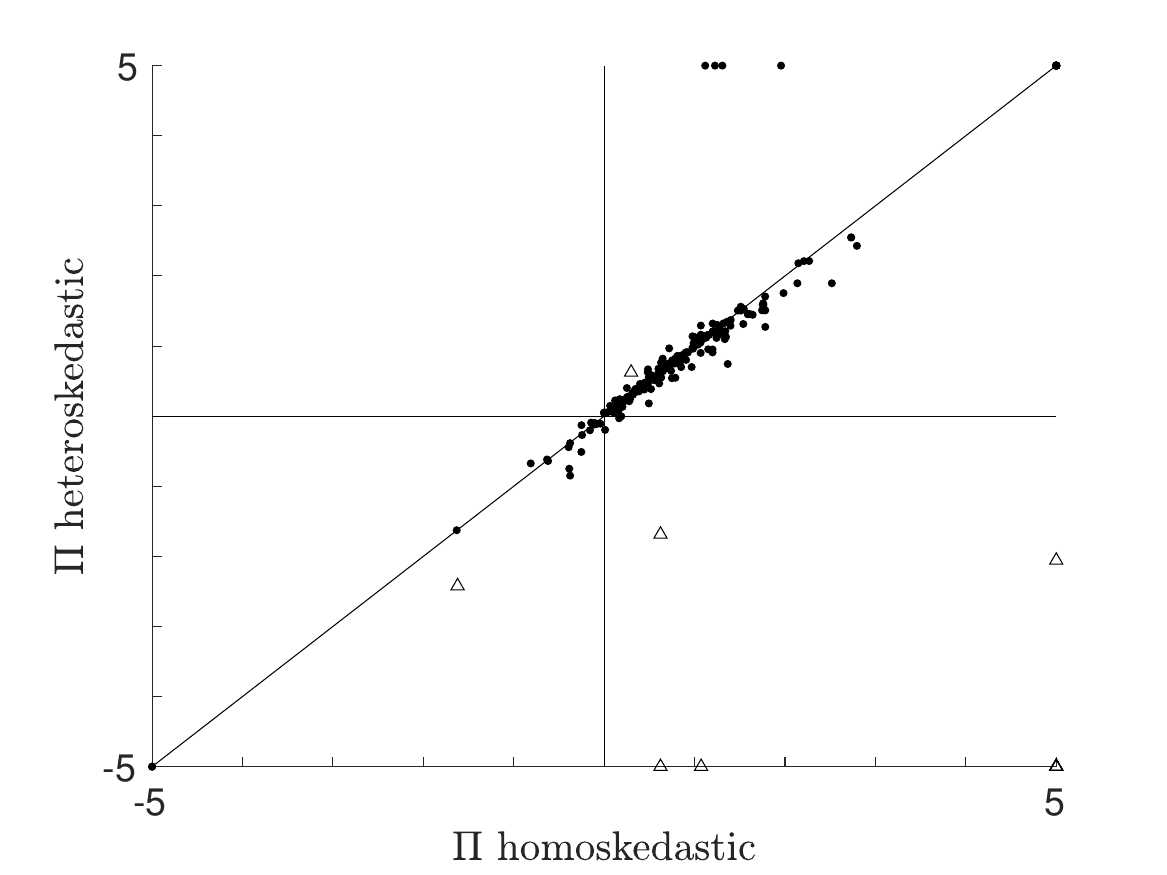}
         \caption{}
         \label{fig:pi_gamma_homo_hetero}
     \end{subfigure}
     \hfill
     \begin{subfigure}[b]{0.45\textwidth}
         \centering
         \includegraphics[width=\textwidth]{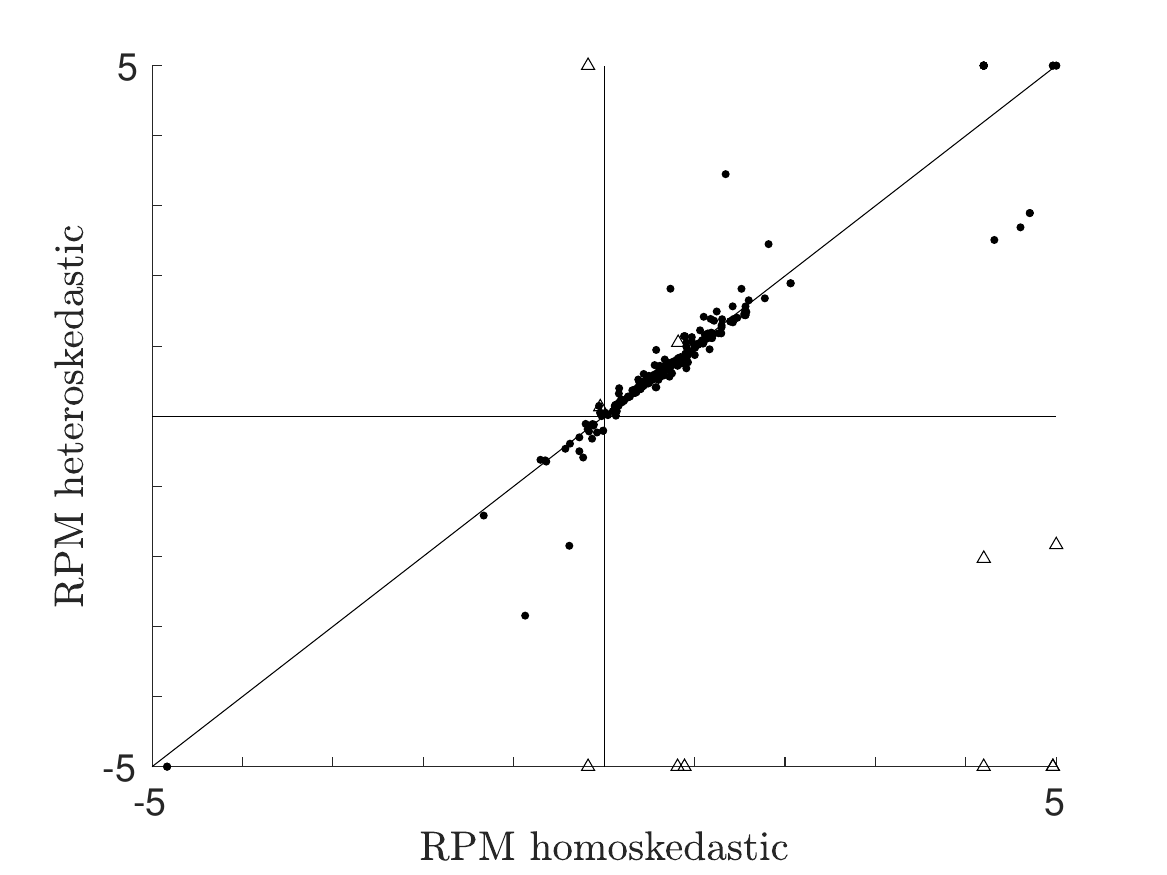}
         \caption{}
         \label{fig:rpm_gamma_homo_hetero}
     \end{subfigure}
      \begin{subfigure}[b]{0.45\textwidth}
         \centering
         \includegraphics[width=\textwidth]{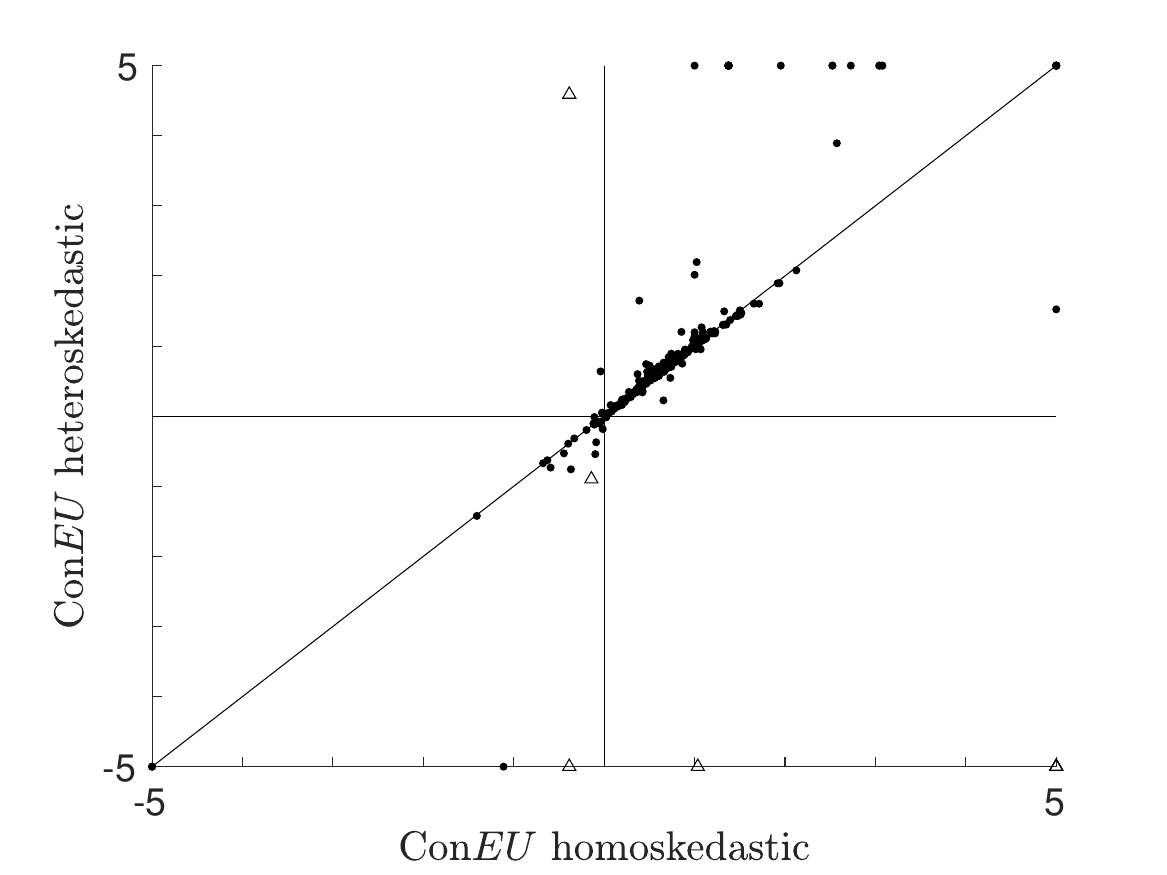}
         \caption{}
         \label{fig:coneu_gamma_homo_hetero}
     \end{subfigure}
        \caption{Homoskedastic vs. heteroskedastic estimates of $\gamma_i$ for each model. $\gamma_i$ estimates outside the range $[-5,5]$ are projected onto the boundary.}
        \label{fig:all_gamma_homo_hetero}
\end{figure}

For all models except for the $EU$-based model, we thus observe -- within the boundaries of what could be expected -- a good match between the homoskedastic and heteroskedastic specifications. This observation is good news about both specifications. On the one hand, we see that the assumption of homoskedasticity does not lead to strong distortions in the distributions of estimated $\gamma_i$. On the other hand, we see that the distribution of the estimated $\gamma_i$ is not impacted strongly by the fact that in the heteroskedastic model we estimate three parameters per individual. 

In contrast, for the $EU$-based RUM depicted in panel (a), we observe a stark difference between the risk aversion levels estimated in the homoskedastic and heteroskedastic specifications. In the homoskedastic model, the estimated levels of risk aversion are all squashed between -1 and 1. This leads to a downwards distortion in the estimated $\gamma_i$ for all participants whose risk aversion is above 1. Accordingly, we find a difference of more than 0.2 in estimated risk aversion levels between the homoskedastic and heteroskedastic specifications for 	all but 89 participants. 

We thus conclude that out of the ten models shown  in Figure \ref{fig:all_gamma_homo_hetero} the one model that really should be avoided is the homoskedastic version of the $EU$-based RUM, in line with the argument provided by \citet{barseghyan2018estimating}. There are some natural follow-up questions:  Are there major differences between the remaining nine models? And is any one of them better than the others?

\begin{table}[h!]
\centering
\begin{tabular}{|c|c|c|c|c|c|}
\hline
& $EU$ & CE & $\Pi$ & RPM & Con$EU$ \\ \hline
Heteroscedastic & 99& 60& 104& 102& 68 \\ \hline
Homoscedastic & 32& 63& 47& 51& 60 \\ \hline
\end{tabular}
\caption{For each model specification, the table displays the number of subjects for whom this specification leads to the maximal likelihood under, respectively, the heteroskedastic and homoskedastic models.}
\label{tab:max_likelihood_count}
\end{table}

\begin{figure}[h!]
     \centering
     \begin{subfigure}[b]{0.45\textwidth}
         \centering
         \includegraphics[width=\textwidth]{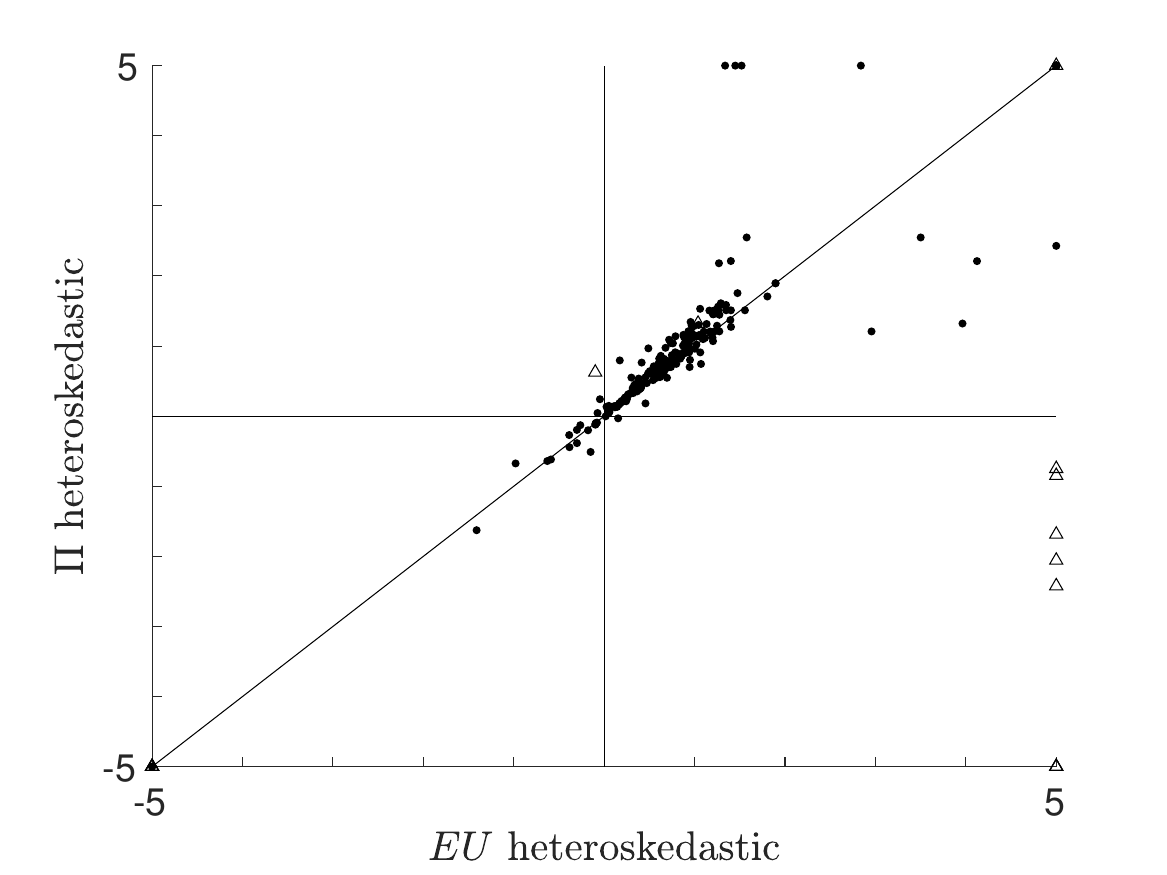}
         \caption{}
         \label{fig:eu_pi_gamma_hetero}
     \end{subfigure}
     \hfill
     \begin{subfigure}[b]{0.45\textwidth}
         \centering
         \includegraphics[width=\textwidth]{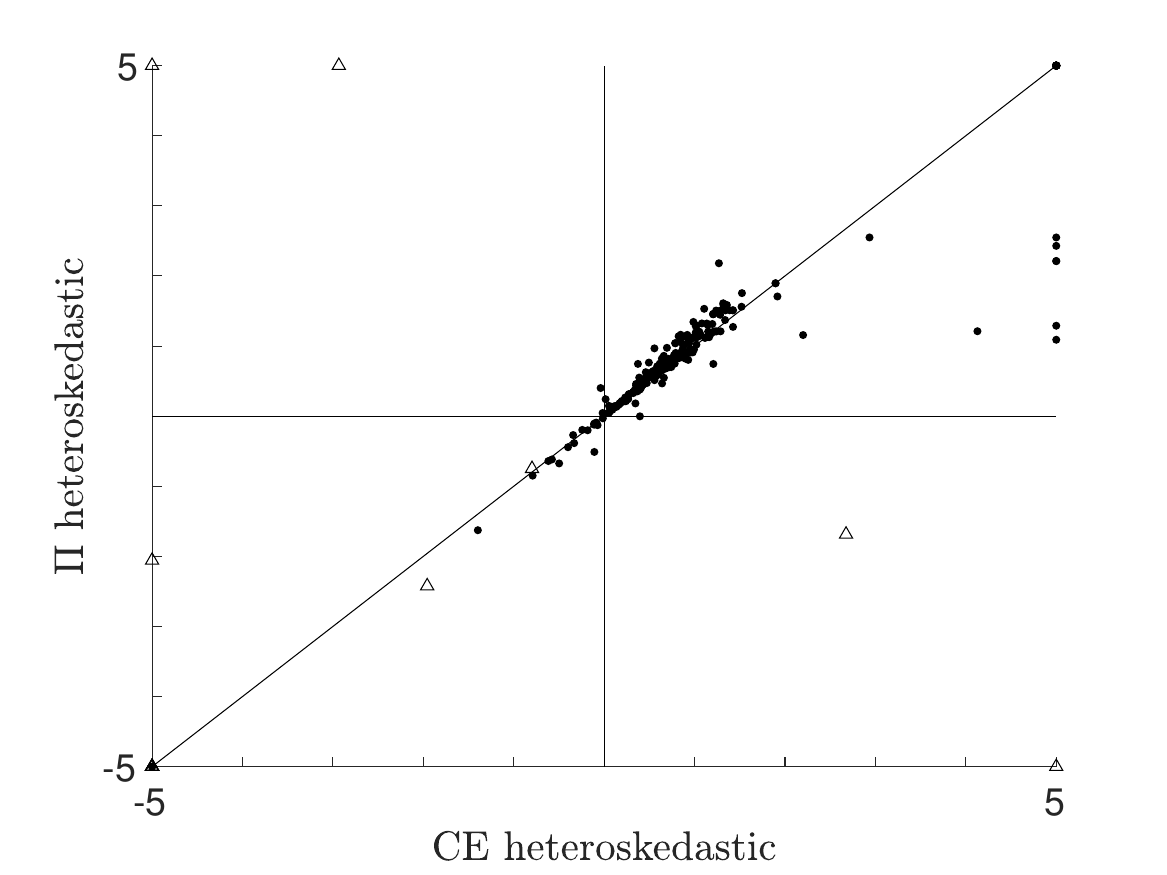}
         \caption{}
         \label{fig:ce_pi_gamma_hetero}
     \end{subfigure}
      \begin{subfigure}[b]{0.45\textwidth}
         \centering
         \includegraphics[width=\textwidth]{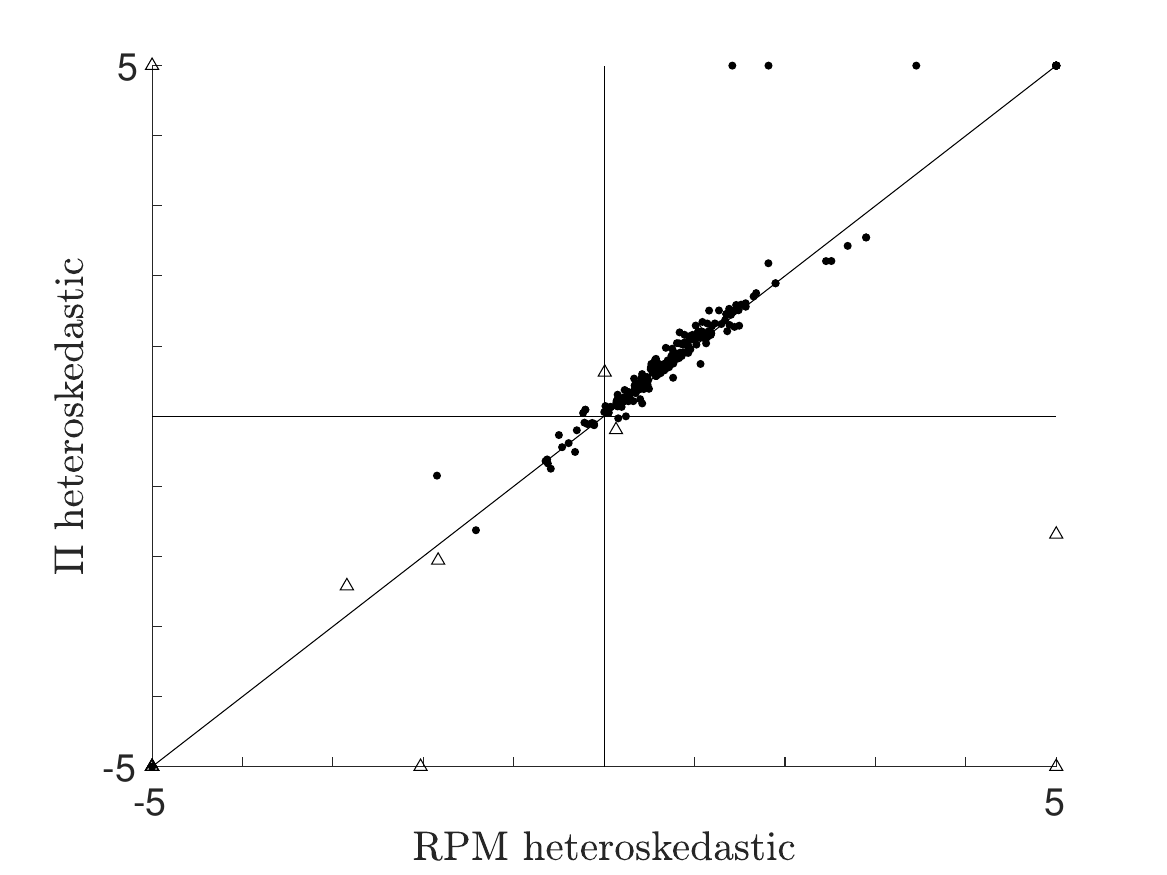}
         \caption{}
         \label{fig:rpm_pi_gamma_hetero}
     \end{subfigure}
     \hfill
     \begin{subfigure}[b]{0.45\textwidth}
         \centering
         \includegraphics[width=\textwidth]{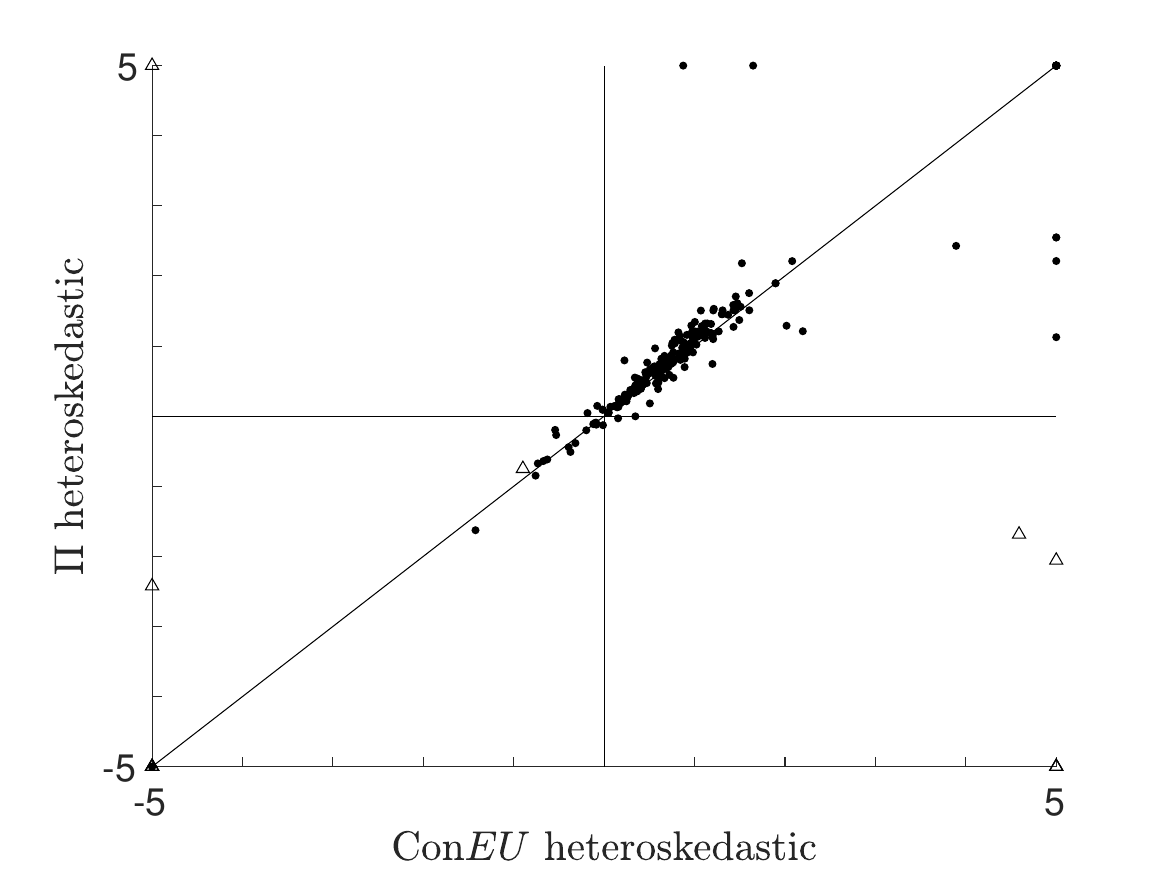}
         \caption{}
         \label{fig:con_eu_pi_gamma_hetero}
     \end{subfigure}
        \caption{Comparison of heteroscedastic estimates of $\gamma_i$ for each model compared to the $\Pi$-based RUM.  $\gamma_i$ estimates outside the range $[-5,5]$ are projected onto the boundary.}
        \label{fig:all_pi_gamma_hetero}
\end{figure}

Figure \ref{fig:all_pi_gamma_hetero} gives a first piece of the answer. The four scatter plots compare the $\gamma_i$ from the heteroskedastic $\Pi$-based RUM to those from the other four heteroskedastic models. We see a good agreement between the five models, especially in the range from -1.84 to 2.21 where the identification of $\gamma_i$ relies on more than the functional form of the error probabilities. This confirms two things: First, all five models seem to be doing more or less the same thing. Second, in the disagreement between the homoskedastic and the heteroskedastic $EU$-based RUM in Figure \ref{fig:eu_gamma_homo_hetero}, it is really the homoskedastic model that differs strongly from the other nine specifications while the heteroskedastic version is in line with models like the $\Pi$-based RUM or the RPM.  

We largely leave the question whether one of the models' empirical performance is a  bit better than the others to future work. Yet, as another small piece of that puzzle, Table \ref{tab:max_likelihood_count} shows for each of the 253 individuals which of the five models leads to the largest likelihood contribution. In the heteroskedastic case, the model that fits the data best for most individuals is our $\Pi$-based RUM, closely followed by the RPM and the $EU$-based RUM. In contrast, in the homoskedastic case, the contextual utility model performs best for the largest group of subjects, followed by the CE-based RUM. None of the five models seems to clearly outperform the others.\footnote{The sum of entries for the heteroskedastic models is larger than for the homoskedastic models due to ties where more than one model maximizes the likelihood of an individual. For instance, there are 38 individuals whose behavior is consistent with maximizing expected CRRA utility and not making mistakes -- a behavior that all models can generate. Other cases where multiple or even all heteroskedastic models reach the same maximal likelihood include, e.g., the six individuals discussed in Footnote \ref{3people}.}

\subsection{The price of non-monotonicity}
Out of the five models under investigation, the $\Pi$-based RUM and the RPM are $\Pi$-monotonic, implying that choice probabilities are monotonic in $\gamma$ for all lottery pairs in the data since those are all $\Pi$-ordered as discussed in Section \ref{DisCE}. While the contextual utility model is not $\Pi$-monotonic as argued in Remark \ref{RemConEU}, one can easily verify that its choice probabilities are monotonic in this example. The remaining two models, the $EU$-based and the CE-based RUM exhibit violations of monotonicity. In the case of the $EU$-based RUM, some of the consequences could be seen clearly in the comparison of the homoskedastic and the heteroskedastic model specifications. In contrast, the impact of non-monotonicity on the performance of the CE-based model is much less evident. One main reason is that in the $EU$-based RUM two problems come together: The differences $V_Y(\gamma)-V_X(\gamma)$ are not just non-monotonic, they also have completely different orders of magnitude depending on the value of $\gamma$. The other four models map the quantities that determine the choice probabilities onto comparable scales, such as the monetary scale for the CE-based and $\Pi$-based RUM, the interval $[0,1]$ for the contextual utility model or coefficients of relative risk aversion in the RPM. In contrast, in the $EU$-based RUM there is no such attempt at normalization. 
 
While the impact of non-monotonicity on the estimates from the CE-based RUM is more subtle, that does not mean that it is not there. Imagine a data set where we observe a few repeated choices from a single lottery pair, e.g. a pair of insurance contracts, in a cross-section of  individuals. This type of empirical setting is typical of data collected ``in the field'' rather than in a controlled laboratory experiment. 
If we now tried to match individual choice probabilities,  a monotonicity failure like in Figure \ref{fig6} would immediately generate an identification problem with multiple values of $\gamma$ leading to exactly the same choice probabilities. 

In our data set, we observe at least 24 different lottery pairs per individual and could thus hope that the identification problems that exist for any given pair of lotteries are washed out in the overall likelihood which combines information from all the different lotteries. To a certain extent, this explains why the CE-based RUM performs fairly well. Nevertheless, the monotonicity violations do not vanish without a trace. Instead, they are a source of multimodality in the likelihoods which can lead to computational problems in the optimization and to counter-intuitive maximum likelihood estimates. This is illustrated in Figure \ref{fig:likelihoods}. The left panel of the figure shows the log-likelihood of one particular participant in the estimation of the homoskedastic CE-based RUM, plotted as a function of $\gamma_i$ at the optimal values of $\lambda$ and $\kappa$. This likelihood function is clearly multimodal with a local maximum near 1, close to the maximum likelihood estimates of the homoskedastic $\Pi$-based RUM and RPM, while the global maximum is at 3.23.\footnote{One can easily show that if all the choice probabilities are monotonic in $\gamma$ -- as is the case for the $\Pi$-based RUM and RPM -- then the likelihood is unimodal in $\gamma_i$ for any fixed $\lambda$ and $\kappa$. Thus, the unimodality in $\gamma$ we see in the figure for those two models is not just a coincidence.} Going back to this individual's choice data, one can see that three blocks of questions were answered consistently with a risk aversion above 1.16 while the switching decision in the remaining block was somewhat less risk-averse, consistent with a risk aversion between 0.15 and 0.65. In this case, the CE-based maximum likelihood estimate of 3.23 seems to be driven mostly by the type of non-monotonic behavior we see in Figure \ref{fig6}.

The right panel of Figure \ref{fig:likelihoods} shows that we also obtain a multimodal likelihood for the pooled model which estimates a single set of parameters $\gamma$, $\lambda$ and $\kappa$ from all  7928 choices made by the 253 subjects across the 40 lotteries. While the global maximum is in line with the global maxima of the RPM and the $\Pi$-based RUM in this case,  the likelihood of the CE-based RUM is \textit{increasing} in $\gamma$ for larger values of $\gamma$ -- which could easily fool a naive attempt at numerically maximizing this likelihood.

\begin{figure}[h!]
     \centering
     \begin{subfigure}[b]{0.45\textwidth}
         \centering
         \includegraphics[width=\textwidth]{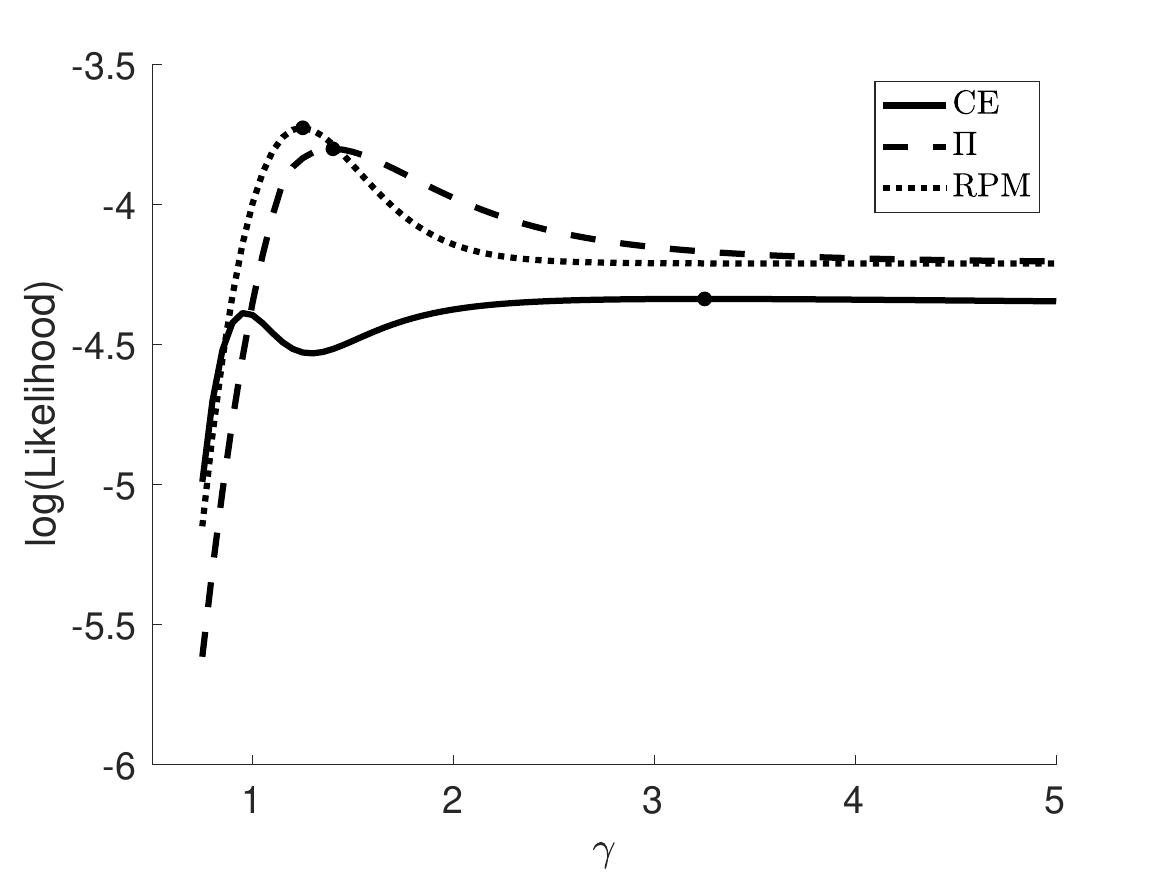}
         \caption{}
         \label{fig:hom_likelihood_ind_133}
     \end{subfigure}
     \hfill
     \begin{subfigure}[b]{0.45\textwidth}
         \centering
         \includegraphics[width=\textwidth]{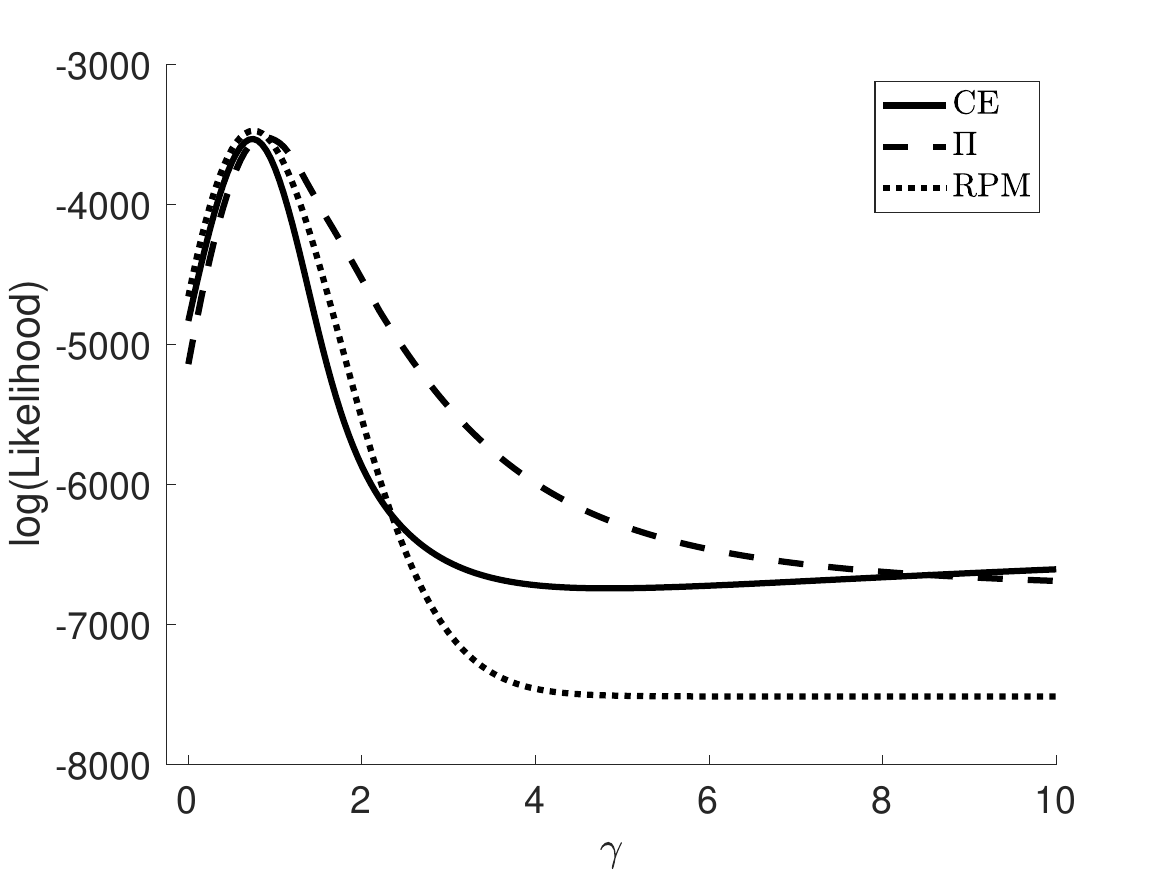}
         \caption{}
         \label{fig:pooled_likelihood}
     \end{subfigure}
        \caption{Log-likelihoods for the CE-based RUM, $\Pi$-based RUM and RPM as a function of $\gamma$. The left panel shows the log-likelihood contribution of the individual with ID 133 from the homoskedastic model while the right panel shows the log-likelihood from the pooled model. In all curves, the parameters $\kappa$ and $\lambda$ are set to the respective optimizers. }
        \label{fig:likelihoods}
\end{figure}

\section{Conclusion}\label{sec7}

The goal of this paper was to emphasize that not all RUMs have bad monotonicity properties when it comes to elicitation of risk preferences. Some RUMs do suffer from monotonicity problems. In particular, the popular $EU$-based RUM should be applied with caution. It should not be combined with a homoskedastic specification of the error term. However, for other RUMs some pessimistic conclusions in earlier research were based on monotonicity criteria that were perhaps overly ambitious, demanding monotonicity for all mean-preserving spreads and more. In this paper, we have argued that the set of lotteries for which one demands monotonicity should be narrowed down -- in line with some classical observations from the literature on choice under background risk. Restricting the set of good lotteries in this way enables us to broaden the set of good stochastic choice models to include again some RUMs. Having more flexibility in the specification of stochastic choice models seems desirable  given that people might make different mistakes in different situations. After all, the RPM, the  CE-based RUM and the $\Pi$-based RUM correspond to different models of how mistakes arise in decision-making: Are mistakes mostly due to uncertainty about one's own preferences? Then an RPM may be the best model for what is going on.  Do agents transform the values of all lotteries into monetary units to make a comparison? Or do they try to quantify the differences between lotteries in monetary terms? The answers may determine whether the  CE-based RUM or the $\Pi$-based RUM is more appropriate. Once it comes to empirical work, we have to live with the mistakes people actually make, not with those they should be making -- and possibly adjust \textit{both} the model and the experiment to circumstances.

\appendix
\section{Proofs}\label{AppA}

\begin{proof}[Proof of Lemma \ref{CElimit}]
The result follows immediately from our assumption that $\CE_\theta(X)$ converges to $\CE_\infty(X)$ and from the fact that the infinitely risk-averse agent is indifferent between two lotteries if and only they have the same worst-possible outcome.
\end{proof}

\begin{proof}[Proof of existence and uniqueness in Definition \ref{Defpi}]
By our assumptions on $u_\theta$, for $X$ and $Y$ with $E[u_{\theta}(X)] \leq E[u_{\theta}(Y)]$, the function $h(\pi)=E[u_{\theta}(X+ \pi)] - E[u_{\theta}(Y)]$ is strictly increasing and continuous in $\pi$ with $h(0)\leq 0$. Moreover, for sufficiently large $\pi$, we have $h(\pi)>0$, i.e., it is always possible to shift $X$ so far upwards that it dominates. Thus, the equation $h(\pi)=0$ has a unique solution. 
\end{proof}

\begin{proof}[Proof of Lemma \ref{lempi}]
Existence and uniqueness of $\pi_\theta(X,Y)$ have been settled already in the context of Definition \ref{Defpi}. 
The claims about the sign of $\pi_{\theta}(X,Y)$ are immediate from its definition. To see the continuity of $\pi_{\theta}(X,Y)$ in $\theta$, note that we can also write the definition of $\pi_{\theta}(X,Y)$ for $(X,Y)$ and $\theta$ with $E[u_{\theta}(Y)] \geq E[u_{\theta}(X)]$ as $F(\theta,\pi_\theta(X,Y))=0$ where the bivariate function $F$ is defined as $F(\theta,\pi)=\CE_\theta(X+\pi)-\CE_\theta(Y)$. By Assumption \ref{A1}, $F$ is continuous in $\theta$. Moreover, by elementary properties of the certainty equivalent, $F$ is continuous and strictly increasing in $\pi$. Continuity of $\pi_\theta(X,Y)$ now follows from the implicit function theorem for continuous and strictly increasing functions found e.g. in \citet{kumagai1980implicit} or \citet[][Theorem 1.H.3]{dontchev2009implicit}.\footnote{Both results are far more general than the result for bivariate functions we need here.} Continuity in the case $E[u_{\theta}(Y)] \leq E[u_{\theta}(X)]$  follows from $\pi_\theta(X,Y)=-\pi_\theta(Y,X)$, noting that we have already shown continuity of  $\pi_\theta(Y,X)$ (under reversed roles of $X$ and $Y$). Finally, in the limit $\theta\rightarrow \infty$, $F(\theta,\pi)$ converges to the strictly increasing and continuous limit $F_\infty(\pi)=\min_i x_i +\pi -\min_j y_j$, implying that there would be a contradiction if $\pi_\theta(X,Y)$ would not converge to $\min_j y_j-\min_i x_i$.
\end{proof}

\begin{proof}[Proof of Lemma \ref{OmegaPi}]
If such a pair of $\theta_1$ and $\theta_2$ did not exist then the pair $(X,Y)$ would actually be $\Pi$-ordered which is a contradiction.
\end{proof}

\begin{proof}[Proof of Lemma \ref{lem4912}]
Denote by $\tilde{X}$ the lottery that takes values $9$ and 12 with equal probability while $\tilde{Y}$ denotes the lottery which always takes the value 10. We know that $\CE_\theta(\tilde{Y})=10$ for all $\theta$ while $\CE_\theta(\tilde{X})$ is strictly decreasing in $\theta$ by \eqref{PR} and Theorem 1 in \citet{pratt1964}. Since $\CE_0(\tilde{X})=10.5$ and 
$\CE_\infty(\tilde{X})=9$, we know that there is some finite $\tilde{\theta}$ with $\CE_{\tilde{\theta}}(\tilde{X})=\CE_{\tilde{\theta}}(\tilde{Y})$ and $\CE_{\theta}(\tilde{X})<\CE_{\theta}(\tilde{Y})$ for any finite $\theta >\tilde{\theta}$. Since $X$ and $Y$ can be interpreted as compound lotteries that pay $4$ with probability 1/3 and, respectively, $\tilde{X}$ and  $\tilde{Y}$ otherwise, it follows that the preferences for $X$ and $Y$ must be in line with those for $\tilde{X}$ and $\tilde{Y}$. Accordingly, $\CE_{\tilde{\theta}}(X)=\CE_{\tilde{\theta}}(Y)$ and $\CE_{\theta}(X)<\CE_{\theta}(Y)$ for any finite $\theta >\tilde{\theta}$ as well as $\CE_{\theta}(X)>\CE_{\theta}(Y)$ for any  $\theta <\tilde{\theta}$. Thus, $X$ and $Y$ are $\Omega$-ordered.  Moreover, we have $\lim_{\theta\rightarrow \infty} \CE_{\theta}(X) =\lim_{\theta\rightarrow \infty} \CE_{\theta}(Y)=4$. By continuity of the compensating premia and the facts that they vanish in case of indifference and are positive in case of a strict preference, we see that  $X$ and $Y$ cannot be $\Pi$-ordered because the compensating premia cannot be monotonic. 
\end{proof}

\begin{proof}[Proof of Lemma \ref{Xy}]
As shown in \citet{pratt1964}, the solution $\pi^e(\theta)$ to 
\[
u_\theta(y-\pi^e(\theta))=E[u_\theta(X)]\;\; \Leftrightarrow\;\; \pi^e(\theta)=y-\CE_\theta(X)
\]
is increasing in $\theta$.\footnote{The notation $\pi^e$ is inspired by \citet{kimball1990} who calls this the equivalent risk premium. Pratt's result assumes that $X$ has mean $y$ but this obviously only shifts the level of $\pi^e_\theta$ by a constant that does not depend on $\theta$.} Moreover, by the Lemma in Section 2 of \citet{kimball1990}, monotonicity of $\pi^e_\theta$ implies monotonicity of the solution $\pi^c(\theta)$ to the equation 
\[
u_\theta(y)=E[u_\theta(X+\pi^c(\theta))],
\]
the so-called compensating premium. Our compensating premium $\pi_\theta(X,Y)$ coincides with $\pi^c(\theta)$ if $u_\theta(y) \geq E[u_\theta(X)]$ and with $\pi^e(\theta)$ if $u_\theta(y) \leq E[u_\theta(X)]$. In the former case, $\pi_\theta(X,Y)$ is non-negative while in the latter case it is non-positive.  Since both expressions are increasing in $\theta$ and since both vanish simultaneously for $u_\theta(y) = E[u_\theta(X)]$, it follows that $\pi_\theta(X,Y)$ is increasing in $\theta$.
\end{proof}

\begin{proof}[Proof of Lemma \ref{kihlstrom}]
We can assume without loss of generality that $E[S]=0$ because otherwise $\pi_\theta({X,Y})$ is just shifted by a constant. The Theorem in \citet{kihlstrom1981} shows that because the functions $u_\theta$ are ranked in terms of their coefficients of absolute risk aversion and because they exhibit decreasing absolute risk aversion, the so-called indirect utility functions $$v_\theta(\cdot)=E[u_\theta(Y+\cdot)]$$ are also ranked in terms of their coefficients of absolute risk aversion, i.e. $-v_\theta''(x)/v_\theta'(x)$ is decreasing in $\theta$. Following the argument in \citet{kihlstrom1981}, we thus conclude that the premium $\pi^e(\theta)$ defined as the solution to
\[
E[u_\theta(Y-\pi^e(\theta))]=v_\theta(-\pi^e(\theta))=E[v_\theta(S)]=E[u_\theta(Y+X)]
\]
is increasing in $\theta$ by Theorem 1 in \citet{pratt1964}. Moreover, by the Lemma about the connection between equivalent and compensating risk premia from \citet{kimball1990}, we conclude that the compensating risk premium defined as the solution to 
\[
v_\theta(0)=E[v_\theta(S+\pi^c(\theta))]
\]
is increasing in $\theta$ as well. The step from monotonicity of $\pi^e(\theta)$ and $\pi^c(\theta)$ to monotonicity of $\pi_{\theta}(X,Y)$ follows by the same arguments as in the proof of Lemma \ref{Xy}.
\end{proof}

\begin{proof}[Proof of Proposition \ref{binaryProp}]
Define the functions 
\[
F^c(\pi,\theta)= p u_\theta(a+\pi)+ (1-p) u_\theta(b+\pi) - p u_\theta(c)- (1-p) u_\theta(d)
\]
and
\[
F^e(\pi,\theta)= p u_\theta(a)+ (1-p) u_\theta(b) - p u_\theta(c-\pi)- (1-p) u_\theta(d-\pi)
\]
and define the premia $\pi^c(\theta)$ and $\pi^e(\theta)$ as the solutions to the equations $F^c(\pi^c(\theta),\theta)=0$ and $F^e(\pi^e(\theta),\theta)=0$. Thus, $\pi_\theta(X,Y)=\pi^c(\theta)$ whenever $Y\succeq_\theta X$ and, thus, $\pi^c(\theta)\geq 0$; and $\pi_\theta(X,Y)=\pi^e(\theta)$ whenever $X\succeq_\theta Y$ and, thus, $\pi^e(\theta)\leq 0$. Similarly to the proof of Lemma \ref{Xy}, it now suffices to show that $\pi^c(\theta)$ and $\pi^e(\theta)$ are increasing in $\theta$ in these cases. By the strict monotonicity of $u_\theta$, it follows that after shifting one of the lotteries upwards by $\pi^c(\theta)\geq 0$ or $-\pi^e(\theta)\geq 0$ the outcomes of the less risky lottery must lie strictly in between the outcomes of the more risky lottery -- otherwise one lottery would dominate the other, contradicting indifference, so, respectively,
\begin{equation}
a+ \pi^c(\theta) < c < d< b+ \pi^c(\theta)
\end{equation}
and 
\begin{equation}\label{abcdtilde}
 a  < c - \pi^e(\theta) < d - \pi^e(\theta)< b
\end{equation}
and, in particular, $\pi^e(\theta)$ and  $\pi^c(\theta)$ always exist because, e.g., equalizing $a+ \pi^c(\theta) = c$ would let the shifted $X$ dominate $Y$ while setting  $d= b+ \pi^c(\theta)$ would let $Y$ dominate the shifted $X$. Let us now study the monotonicity behavior of $\pi^e(\theta)$ under CRRA utility, so $\theta$ becomes $\gamma$, for $\gamma\neq 1$. The result for all $\gamma$ then follows by continuity of the risk premium. We thus apply the usual implicit function theorem for smooth functions to 
\begin{align}\label{Fe}
0&\stackrel{!}{=}F^e(\pi^e(\gamma),\gamma)\nonumber\\
&= p a^{1-\gamma}+ (1-p) b^{1-\gamma} - p (c-\pi^e(\gamma))^{1-\gamma}- (1-p) (d-\pi^e(\gamma))^{1-\gamma} 
\end{align}
where we dropped the multiplicative factor $\frac{1}{1-\gamma}$, thus slightly altering the definition of $F^e$. The derivative we need is thus given by 
\begin{equation}\label{dpiFe}
\frac{d\pi^e(\gamma)}{d \gamma}=- \frac{\frac{\partial}{\partial \gamma} F^e(\pi^e(\gamma),\gamma)}{\frac{\partial}{\partial \pi} F^e(\pi^e(\gamma),\gamma)}.
\end{equation}
The derivative in the denominator is given by 
\[
\frac{\partial}{\partial \pi} F^e(\pi,\gamma) = (1-\gamma)\left(
p(c-\pi)^{-\gamma} + (1-p)(d-\pi)^{-\gamma}
\right)
\]
which is always non-zero for $\pi=\pi^e(\gamma)$ and has the same sign as $(1-\gamma)$. The derivative in the numerator is 
\begin{align*}
\frac{\partial}{\partial \gamma} F^e(\pi,\gamma) =& 
-p a^{1-\gamma}\log(a)- (1-p) b^{1-\gamma}\log(b)\\ &+ p (c-\pi)^{1-\gamma}\log(c-\pi)+ (1-p) (d-\pi)^{1-\gamma} \log(d-\pi). 
\end{align*}
so that 
\begin{align*}
&(1-\gamma)\frac{\partial}{\partial \gamma} F^e(\pi,\gamma) =
-p a^{1-\gamma}\log(a^{1-\gamma})- (1-p) b^{1-\gamma}\log(b^{1-\gamma})\\ &+ p (c-\pi)^{1-\gamma}\log( (c-\pi)^{1-\gamma})+ (1-p) (d-\pi)^{1-\gamma} \log( (d-\pi)^{1-\gamma}). 
\end{align*}
Now, define $\tilde{a}=a^{1-\gamma}$, $\tilde{b}=b^{1-\gamma}$, $\tilde{c}=(c-\pi^e(\gamma))^{1-\gamma}$ and $\tilde{d}=(d-\pi^e(\gamma))^{1-\gamma}$ as well as $h(z)=z \log(z)$ and note that we can write  
\begin{align}\label{abcdineq}
&(1-\gamma)\frac{\partial}{\partial \gamma} F^e(\pi^e(\gamma),\gamma) =
-\left(p h(\tilde{a})+ (1-p) h(\tilde{b})\right)+ \left(p h(\tilde{c})+ (1-p) h(\tilde{d})\right). 
\end{align}
Next, note that $h:\mathbb{R}^+ \rightarrow \mathbb{R}$ is a convex function and that the numbers $\tilde{a}$, $\tilde{b}$, $\tilde{c}$ and  $\tilde{d}$ satisfy $p \tilde{a}+(1-p) \tilde{b} = p\tilde{c}+(1-p)\tilde{d}$ by \eqref{Fe}, as well as,  $\tilde{a}>\tilde{c}>\tilde{d}>\tilde{b}$ or $\tilde{a}<\tilde{c}<\tilde{d}<\tilde{b}$, depending on the sign of $1-\gamma$. In either case, it now follows from elementary properties of convex functions that $(1-\gamma)\frac{\partial}{\partial \gamma} F^e(\pi,\pi^e(\gamma))  <0$. Namely,  we can interpret the right hand side of \eqref{abcdineq} as comparing the line segment connecting the points $(\tilde{a}, h(\tilde{a}))$ and  $(\tilde{b}, h(\tilde{b}))$ to the line segment connecting  $(\tilde{c}, h(\tilde{c}))$ and  $(\tilde{d}, h(\tilde{d}))$ when evaluated at 
$\tilde{m}:=p \tilde{a}+(1-p) \tilde{b} = p\tilde{c}+(1-p)\tilde{d}$. Since the numbers $\tilde{a}$ and  $\tilde{b}$ lie further apart than  $\tilde{c}$ and  $\tilde{d}$, it follows from the convexity of $h$ that the right hand side of \eqref{abcdineq} is negative. Thus, $\frac{\partial}{\partial \gamma} F^e(\pi,\pi^e(\gamma))$ always has the opposite sign of $1-\gamma$. 

In \eqref{dpiFe}, we thus have minus the ratio of a quantity with the same sign as $1-\gamma$ and a quantity with the opposite sign as $1-\gamma$ which gives a positive number. It follows that $\pi^e(\gamma)$ is increasing in $\gamma$, implying that $\pi_\gamma(X,Y)=\pi^e(\gamma)$ is increasing in $\gamma$ as well. 

We close the proof by showing that $\pi^c(\alpha)$ is increasing in $\alpha\neq 0$ in the case of CARA utility. The remaining two cases, i.e., monotonicity of $\pi^e$ for CARA and of $\pi^c$ for CRRA then follow analogously. In this case, 
\begin{align*}\label{Fc}
0&\stackrel{!}{=}F^c(\pi^c(\alpha),\gamma)\nonumber\\
&=  p e^{-\alpha (a+\pi^c(\gamma))}+ (1-p) e^{-\alpha (b+\pi^c(\gamma))} - p e^{-\alpha c}- (1-p)e^{-\alpha d} 
\end{align*}
where we have dropped the multiplicative factor $-\frac{1}{\alpha}$, thus slightly altering the definition of $F^c$. We now need to compute
\begin{equation}\label{dpiFc}
\frac{d\pi^c(\alpha)}{d \alpha}=- \frac{\frac{\partial}{\partial \alpha} F^c(\pi^c(\alpha),\alpha)}{\frac{\partial}{\partial \pi} F^c(\pi^c(\alpha),\alpha)}.
\end{equation}
The derivative in the denominator is given by 
\[
\frac{\partial}{\partial \pi} F^c(\pi,\alpha) =  -\alpha \left(p e^{-\alpha (a+\pi^c(\gamma))}+ (1-p) e^{-\alpha (b+\pi^c(\gamma))}\right)
\]
which always has the opposite sign of $\alpha$. The derivative in the numerator is given by 
\[
\frac{\partial}{\partial \alpha} F^c(\pi,\alpha) =
-p (a+\pi)e^{-\alpha (a+\pi)}- (1-p)(b+\pi) e^{-\alpha (b+\pi)} + p ce^{-\alpha c}+ (1-p)de^{-\alpha d}. 
\]
Now, defining $\tilde{a}=e^{-\alpha (a+\pi^c(\alpha))}$, $\tilde{b}=e^{-\alpha (b+\pi^c(\alpha))}$, $\tilde{c}=e^{-\alpha c}$ and  $\tilde{d}=e^{-\alpha d}$ as well as, again, $h(z)=z \log(z)$ we can write 
\[
\alpha \frac{\partial}{\partial \alpha} F^c(\pi^c(\alpha),\alpha)= \left(p h(\tilde{a})+ (1-p) h(\tilde{b})\right)- \left(p h(\tilde{c})+ (1-p) h(\tilde{d})\right).
\]
As in the CRRA case, we have $p \tilde{a}+(1-p) \tilde{b} = p\tilde{c}+(1-p)\tilde{d}$  and either  $\tilde{a}>\tilde{c}>\tilde{d}>\tilde{b}$ or $\tilde{a}<\tilde{c}<\tilde{d}<\tilde{b}$, depending on the sign of $\alpha$. It follows that $\alpha \frac{\partial}{\partial \alpha} F^c(\pi^c(\alpha),\alpha)>0$ so $\alpha$ and $\frac{\partial}{\partial \alpha} F^c(\pi^c(\alpha),\alpha)$ always have the same sign. Putting the pieces together, we find that the right hand side in \eqref{dpiFc} is always positive, implying that $\pi^c(\alpha)$ is increasing in $\alpha$.
\end{proof}

\begin{proof}[Proof of Lemma \ref{lemtop}]
This is a direct consequence of the fact that
\[
E[u_\theta(X_1+\pi_\theta(X_1,X_2))] = E[u_\theta(X_2)]=E[u_\theta(X_3)].
\]
\end{proof}

\begin{proof}[Proof of Proposition \ref{PiRUMPi}]
Let $(X,Y)$ be a $\Pi$-ordered pair so that $\pi_\theta({X,Y})$ is increasing in $\theta$. We need to show that the probability of choosing $X$ decreases in $\theta$. Denote by $V_X(\theta)=-\pi_\theta({X,X_{\max}})$ and $V_Y(\theta)=-\pi_\theta({Y,X_{\max}})$ the respective preference indices. If $E[u_\theta(X)]\geq E[u_\theta(Y)]$ then $X_{\max}=X$ and $V_Y(\theta)-V_X(\theta)=-\pi_\theta({Y,X})+\pi_\theta({X,X})=\pi_\theta({X,Y})$. If $E[u_\theta(X)]< E[u_\theta(Y)]$ then $X_{\max}=Y$ and $V_Y(\theta)-V_X(\theta)=-\pi_\theta({Y,Y})+\pi_\theta({X,Y})=\pi_\theta({X,Y})$. Thus, $V_Y(\theta)-V_X(\theta)=\pi_\theta(X,Y)$. Accordingly, by \eqref{genRUMXY}, we can write the probability of choosing $Y$ as 
\[
\mathbb{P}(V_Y+\xi_Y/\lambda > V_X+\xi_X/\lambda )=\Phi\left(\lambda \pi_{X,Y}(\theta)\right)
\]
where $\Phi$ denotes the cumulative distribution function of $\xi_X-\xi_Y$ which is strictly increasing by our continuous distribution assumption. Thus the probability of choosing $Y$ is increasing in $\theta$ and the probability of choosing $X$ is decreasing.
\end{proof}

\begin{proof}[Proof of Proposition \ref{CARAmon}]
By Lemma \ref{CARApi}, we know that with CARA utility $\pi_\theta(X_i,X_{\max}(\theta)) = \CE_\theta(X_{\max}(\theta))-\CE_\theta(X_i)
$ and thus $V_i(\theta)=\CE_\theta(X_i)-\CE_\theta(X_{\max}(\theta))$ holds for the $\Pi$-based RUM. Clearly, the choice probabilities are not changed if we replace $V_i(\theta)$ by $\tilde{V}_i(\theta)=V_i(\theta)+C$ for some constant $C$ that does not depend on $i$. Choosing $C=\CE_\theta(X_{\max}(\theta))$ gives $\tilde{V}_i(\theta)= \CE_\theta(X_i)$ which corresponds exactly to the CE-based RUM.
\end{proof}

\begin{proof}[Proof of Proposition \ref{onlyCara}]
For any gamble $X$ and real number $c$ such that $X+c$ only takes values in $\S$, we have $\pi_\theta({X,X+c})=c$ for all $\theta$, implying that $(X,X+c)$ is a $\Pi$-ordered pair. Moreover, we have that $\CE_\theta(X+c)-\CE_\theta(X)=c$ for $\theta \in \{0,\infty\}$. Monotonicity of  $\CE_\theta(X+c)-\CE_\theta(X)$ for all $\Pi$-ordered pairs now implies that $\CE_\theta(X+c)-\CE_\theta(X)$ must be constant for $\theta\in(0,\infty)$. We have thus shown translation invariance, $\CE_\theta(X+c)-\CE_\theta(X)=c$ for all $X$ and $c$. Yet,  translation invariance does not hold for any non-CARA utility functions, see e.g. Theorem 2.2 in \citet{Mueller}. Thus, all the functions $u_\theta$ must be CARA.
\end{proof}

\begin{proof}[Proof of Proposition \ref{PiRUMprop}]
Continuity of the premia $\pi_\theta({X_i,X_j})$ in $\theta$ has been established in Lemma \ref{lempi}. Regarding continuity of $\pi_\theta({X_i,X_{\max}})$, the only remaining worry is then that a discontinuity arises when multiple choices are tied at the top for some $\theta$. Lemma \ref{lemtop} has shown that this is not a problem. Properties \eqref{vanishingnoise} and \eqref{topordering} now both follow from the observation that $V_i(\theta)=0$ if $E[u_\theta(X_i)]\geq E[u_\theta(X_j)]$ for all $j$ and $V_i(\theta)<0$ otherwise, implying that gambles with $V_i(\theta)=0$ have the highest probability of being chosen, and that the choice probabilities of all gambles with $V_i(\theta)<0$ must vanish as $\lambda \rightarrow \infty$. 
\end{proof}

\section{Details on the empirical illustration from Section \ref{empirics}}\label{AEmpirics}

Denote by $L_i$ the set of 24 or 40 lottery pairs $(X_j,Y_j)$ offered to participant $i$. For a lottery $j\in L_i$, we denote by $D_{ij}^X$, $D_{ij}^Y$ and $D_{ij}^I$ three dummy variables that take the value 1 if participant $i$ chooses respectively, option $X_j$, option $Y_j$ or indifference for this lottery pair; and the value 0 otherwise. We denote by $\rho_{\gamma_i,\lambda_i}(X_j)$ the probability of choosing $X_j$ from the pair $(X_j,Y_j)$ in the baseline version of our choice models without the tremble probabilities. For the four RUMs -- the $EU$-based, CE-based and $\Pi$-based RUM as well as contextual utility -- we thus have 
\[
\rho_{\gamma_i,\lambda_i}(X_j) = \frac{\exp(\lambda_i V_{X_j}(\gamma_i))}{\exp(\lambda_i V_{X_j}(\gamma_i))+\exp(\lambda_i V_{Y_j}(\gamma_i))}
\]
where the utility index $V$ corresponds to, respectively, expected utility, the certainty equivalent, the compensating premium or rescaled expected utility. For the RPM, one can show that for the 36 nondegenerate lottery pairs 
\[
\rho_{\gamma_i,\lambda_i}(X_j) =  \frac{\exp(\lambda_i \bar{\gamma}_j)}{\exp(\lambda_i \bar{\gamma}_j)+\exp(\lambda_i \gamma_i)}
\]
where $\bar{\gamma}_j$ denotes the unique value of $\gamma$ that makes a CRRA agent indifferent between $X_j$ and $Y_j$. For the four degenerate lotteries in which $X_j$ dominates $Y_j$, the RPM choice probability is just $\rho_{\gamma_i,\lambda_i}(X_j)=1$.
The overall probability of choosing $X_j$ accounting for the tremble probability $\kappa_i$ is then given by 
\[
\rho_{\gamma_i,\lambda_i,\kappa_i}(X_j) = (1-\kappa_i)\rho_{\gamma_i,\lambda_i}(X_j) + \kappa_i (1-\rho_{\gamma_i,\lambda_i}(X_j))
\]
and the log-likelihood that we are maximizing in the heteroskedastic specification of each model is given by 
\[
\sum_i \sum_{j\in L_i} \left( D_{ij}^X +\frac{D_{ij}^I}2  \right)\log(\rho_{\gamma_i,\lambda_i,\kappa_i}(X_j)) + \left( D_{ij}^Y +\frac{D_{ij}^I}2 \right)\log(1-\rho_{\gamma_i,\lambda_i,\kappa_i}(X_j)).
\]
In the homoskedastic specification, we maximize the same log-likelihood but add the constraints $\lambda_i=\lambda$ and $\kappa_i=\kappa$ for all $i$. Finally, for the pooled model, we add the constraint $\gamma_i=\gamma$. The results of the homoskedastic and heteroskedastic specifications are discussed in detail in the main text. The results of the pooled specification are found in Table \ref{tab:pooled_estimates}. We see that all five models give similar estimates of $\gamma$ with the smallest coming from the $EU$-based RUM and the largest from the $\Pi$-based RUM.

\begin{table}[h!]
\centering
\begin{tabular}{|c|c|c|c|c|c|}
\hline
& $EU$ & CE & $\pi$ & RPM & Con$EU$ \\ \hline
$\gamma$ & \makecell{0.661\\(0.031)} & \makecell{0.738\\(0.034)} & \makecell{0.872\\(0.049)} & \makecell{0.752\\(0.042)} & \makecell{0.676\\(0.033)} \\ \hline
$\lambda$ & \makecell{0.275\\(0.067)} & \makecell{0.002\\(0.000)} & \makecell{0.003\\(0.000)} & \makecell{2.495\\(0.133)} & \makecell{8.598\\(0.426)} \\ \hline
$\kappa$ & \makecell{0.034\\(0.011)} & \makecell{0.026\\(0.010)} & \makecell{0.048\\(0.009)} & \makecell{0.051\\(0.008)} & \makecell{0.027\\(0.010)} \\ \hline
\end{tabular}
\caption{Pooled estimates for $\gamma$, $\lambda$ and $\kappa$ for each model. Block bootstrap standard errors from 10000 repetitions are in parentheses.}
\label{tab:pooled_estimates}
\end{table}

As noted in the main text, multimodality is a very real concern at least for some of our log-likelihoods. We thus rely on a global optimization algorithm which can handle multimodality. To this end, we combine a genetic algorithm with a refinement using a quasi-Newton method that takes the output of the genetic algorithm as a starting point.\footnote{For the genetic algorithm, we rely on the implementation in \textsc{Matlab}'s \texttt{ga} command while the quasi-Newton method is from \textsc{Matlab}'s \texttt{fminunc}.} The constraints $\lambda_i\geq 0$ and $\kappa_i\in [0,1]$ we handle implicitly by optimizing over suitably transformed variables $(\tilde{\lambda}_i,\tilde{\kappa}_i) \in \mathbb{R}^2$ defined via $\lambda_i=\exp(\tilde{\lambda}_i)$ and $\kappa=(1+\exp(-\tilde{\kappa}_i))^{-1}$. For the pooled model, we directly solve the optimization in three variables in this way. For the heteroskedastic model, we have analogous three-variable problems for each individual that can be solved independently and in parallel. Finally, for the homoskedastic model, we apply the combination of genetic algorithm and quasi-Newton method to an outer optimization over $\lambda$ and $\kappa$ while the univariate inner optimizations that pick the $\gamma_i$ for given $\lambda$ and $\kappa$ rely only on the quasi-Newton method for efficiency reasons. 

In order to achieve an efficient estimation of the $\Pi$-based RUM, we precompute the curves $\pi_\gamma(X_j,Y_j)$ for all 40 lotteries on a grid of $\gamma$-values from $-20$ to $20$ in steps of $0.01$. In the optimization, we then interpolate the values on this grid using piecewise cubic Hermite interpolation (pchip) and apply constant extrapolation.\footnote{The justification for this approach is that the function $\pi_\gamma(X_j,Y_j)$ is smooth and increases continuously from its left limit to its right limit, being almost constant for $\gamma$-values in the tails. The limits are given by, respectively, the differences between the best-possible and between the worst-possible outcomes and do not differ much from the values we see at $\pm 20$. Using pchip as the interpolation method has the advantage that we obtain a smooth function which preserves the monotonicity behavior of the input.} 

In the sample, there are 38 participants whose choices are consistent with non-sto\-chas\-tic CRRA utility at some level $\gamma$. In all the heteroskedastic specifications, the maximum likelihood is achieved by setting the error probability to zero, $\lambda_i \rightarrow \infty$ and $\kappa_i=0$, and by choosing \textit{some} $\gamma$-value that lies between the relevant indifference thresholds. In particular, the optimization algorithms for the different models may stop at different values in the relevant range, suggesting disagreement between the  models where there is none. To counteract this impression, we set the $\gamma$-estimates for the 24 individuals whose behavior is consistent with lying between two interior indifference thresholds to the midpoints between those thresholds. For the 14 individuals who lie above the highest or below the lowest threshold, we set their estimates to, respectively, +5 and -5.


\begin{thebibliography}{33}
\providecommand{\natexlab}[1]{#1}
\providecommand{\url}[1]{\texttt{#1}}
\expandafter\ifx\csname urlstyle\endcsname\relax
  \providecommand{\doi}[1]{doi: #1}\else
  \providecommand{\doi}{doi: \begingroup \urlstyle{rm}\Url}\fi

\bibitem[Al{\'o}s-Ferrer and Garagnani(2021)]{alos2021choice}
Carlos Al{\'o}s-Ferrer and Michele Garagnani.
\newblock Choice consistency and strength of preference.
\newblock \emph{Economics Letters}, 198:\penalty0 109672, 2021.

\bibitem[Andersen et~al.(2008)Andersen, Harrison, Lau, and
  Rutstr{\"o}m]{andersen2008eliciting}
Steffen Andersen, Glenn~W. Harrison, Morten~I. Lau, and E.~Elisabet
  Rutstr{\"o}m.
\newblock Eliciting risk and time preferences.
\newblock \emph{Econometrica}, 76\penalty0 (3):\penalty0 583--618, 2008.

\bibitem[Apesteguia and Ballester(2018)]{AB}
Jose Apesteguia and Miguel~A. Ballester.
\newblock Monotone stochastic choice models: The case of risk and time
  preferences.
\newblock \emph{Journal of Political Economy}, 126\penalty0 (1):\penalty0
  74--106, 2018.

\bibitem[Balter et~al.(2024)Balter, Chau, and Schweizer]{balter2024}
Anne~G. Balter, Ki~Wai Chau, and Nikolaus Schweizer.
\newblock Comparative risk aversion vs. threshold choice in the {O}mega ratio.
\newblock \emph{Omega}, 123:\penalty0 102992, 2024.

\bibitem[Barseghyan et~al.(2018)Barseghyan, Molinari, O’Donoghue, and
  Teitelbaum]{barseghyan2018estimating}
Levon Barseghyan, Francesca Molinari, Ted O’Donoghue, and Joshua~C.
  Teitelbaum.
\newblock Estimating risk preferences in the field.
\newblock \emph{Journal of Economic Literature}, 56\penalty0 (2):\penalty0
  501--564, 2018.

\bibitem[Bellemare(2023)]{bellemare2023estimation}
Charles Bellemare.
\newblock \emph{Estimation of Structural Models Using Experimental Data From
  the Lab and the Field}.
\newblock Cambridge University Press, 2023.

\bibitem[Blavatskyy(2011)]{blavatskyy}
Pavlo~R. Blavatskyy.
\newblock Probabilistic risk aversion with an arbitrary outcome set.
\newblock \emph{Economics Letters}, 112\penalty0 (1):\penalty0 34--37, 2011.

\bibitem[Bruhin et~al.(2010)Bruhin, Fehr-Duda, and Epper]{bruhin2010risk}
Adrian Bruhin, Helga Fehr-Duda, and Thomas Epper.
\newblock Risk and rationality: Uncovering heterogeneity in probability
  distortion.
\newblock \emph{Econometrica}, 78\penalty0 (4):\penalty0 1375--1412, 2010.

\bibitem[Buschena and Zilberman(2000)]{buschena2000generalized}
David Buschena and David Zilberman.
\newblock Generalized expected utility, heteroscedastic error, and path
  dependence in risky choice.
\newblock \emph{Journal of Risk and Uncertainty}, 20:\penalty0 67--88, 2000.

\bibitem[Charness et~al.(2013)Charness, Gneezy, and
  Imas]{charness2013experimental}
Gary Charness, Uri Gneezy, and Alex Imas.
\newblock Experimental methods: Eliciting risk preferences.
\newblock \emph{Journal of economic behavior \& organization}, 87:\penalty0
  43--51, 2013.

\bibitem[Conte and Hey(2018)]{conterehabilitating}
Anna Conte and John~D. Hey.
\newblock Rehabilitating the random utility model. {A} comment on {A}pesteguia
  and {B}allester (2018).
\newblock \emph{Working Paper, York University}, 2018.

\bibitem[Dontchev and Rockafellar(2014)]{dontchev2009implicit}
Asen~L. Dontchev and R.~Tyrrell Rockafellar.
\newblock \emph{Implicit functions and solution mappings}.
\newblock Springer, 2nd edition, 2014.

\bibitem[Eckel and Grossman(2002)]{eckel2002sex}
Catherine~C. Eckel and Philip~J. Grossman.
\newblock Sex differences and statistical stereotyping in attitudes toward
  financial risk.
\newblock \emph{Evolution and Human Behavior}, 23\penalty0 (4):\penalty0
  281--295, 2002.

\bibitem[Friedman(1974)]{friedman1974risk}
Bernard Friedman.
\newblock Risk aversion and the consumer choice of health insurance option.
\newblock \emph{Review of Economics and Statistics}, 56\penalty0 (2):\penalty0
  209--214, 1974.

\bibitem[Gaudecker et~al.(2011)Gaudecker, Soest, and
  Wengstr{\"o}m]{von2011heterogeneity}
Hans-Martin~{von} Gaudecker, Arthur~{{van}} Soest, and Erik Wengstr{\"o}m.
\newblock Heterogeneity in risky choice behavior in a broad population.
\newblock \emph{American Economic Review}, 101\penalty0 (2):\penalty0 664--694,
  2011.

\bibitem[Harless and Camerer(1994)]{harless}
David~W. Harless and Colin~F. Camerer.
\newblock The predictive utility of generalized expected utility theories.
\newblock \emph{Econometrica}, 62\penalty0 (6):\penalty0 1251--1289, 1994.

\bibitem[Hey(1995)]{hey1995experimental}
John~D. Hey.
\newblock Experimental investigations of errors in decision making under risk.
\newblock \emph{European Economic Review}, 39\penalty0 (3-4):\penalty0
  633--640, 1995.

\bibitem[Holt and Laury(2002)]{holt2002risk}
Charles~A. Holt and Susan~K. Laury.
\newblock Risk aversion and incentive effects.
\newblock \emph{American Economic Review}, 92\penalty0 (5):\penalty0
  1644--1655, 2002.

\bibitem[Holzmeister and Stefan(2021)]{holzmeister2021risk}
Felix Holzmeister and Matthias Stefan.
\newblock The risk elicitation puzzle revisited: Across-methods (in)
  consistency?
\newblock \emph{Experimental Economics}, 24:\penalty0 593--616, 2021.

\bibitem[Kihlstrom et~al.(1981)Kihlstrom, Romer, and Williams]{kihlstrom1981}
Richard~E. Kihlstrom, David Romer, and Steve Williams.
\newblock Risk aversion with random initial wealth.
\newblock \emph{Econometrica}, 49\penalty0 (4):\penalty0 911--920, 1981.

\bibitem[Kimball(1990)]{kimball1990}
Miles~S. Kimball.
\newblock Precautionary saving in the small and in the large.
\newblock \emph{Econometrica}, 58\penalty0 (1):\penalty0 53--73, 1990.

\bibitem[Kumagai(1980)]{kumagai1980implicit}
Sadatoshi Kumagai.
\newblock An implicit function theorem: Comment.
\newblock \emph{Journal of Optimization Theory and Applications}, 31:\penalty0
  285--288, 1980.

\bibitem[McFadden(2001)]{mcfadden2001economic}
Daniel McFadden.
\newblock Economic choices.
\newblock \emph{American Economic Review}, 91\penalty0 (3):\penalty0 351--378,
  2001.

\bibitem[M{\"u}ller(2007)]{Mueller}
Alfred M{\"u}ller.
\newblock Certainty equivalents as risk measures.
\newblock \emph{Brazilian Journal of Probability and Statistics}, 21\penalty0
  (1):\penalty0 1--12, 2007.

\bibitem[Olschewski et~al.(2022)Olschewski, Sirotkin, and
  Rieskamp]{olschewski2022empirical}
Sebastian Olschewski, Pavel Sirotkin, and J{\"o}rg Rieskamp.
\newblock Empirical underidentification in estimating random utility models:
  The role of choice sets and standardizations.
\newblock \emph{British Journal of Mathematical and Statistical Psychology},
  75\penalty0 (2):\penalty0 252--292, 2022.

\bibitem[Pratt(1964)]{pratt1964}
John~W. Pratt.
\newblock Risk aversion in the small and in the large.
\newblock \emph{Econometrica}, 32\penalty0 (1--2):\penalty0 122--136, 1964.

\bibitem[Ross(1981)]{ross1981}
Stephen~A. Ross.
\newblock Some stronger measures of risk aversion in the small and the large
  with applications.
\newblock \emph{Econometrica}, 49\penalty0 (3):\penalty0 621--638, 1981.

\bibitem[Rothschild and Stiglitz(1970)]{RS}
Michael Rothschild and Joseph Stiglitz.
\newblock Increasing risk: I. {A} definition.
\newblock \emph{Journal of Economic Theory}, 2\penalty0 (3):\penalty0 225--243,
  1970.

\bibitem[Stewart et~al.(2018)Stewart, Scheibehenne, and
  Pachur]{stewart2018psychological}
Neil Stewart, Benjamin Scheibehenne, and Thorsten Pachur.
\newblock Psychological parameters have units: A bug fix for stochastic
  prospect theory and other decision models.
\newblock \emph{OSF preprint}, 2018.

\bibitem[Train(2009)]{train2009}
Kenneth~E. Train.
\newblock \emph{Discrete choice methods with simulation}.
\newblock Cambridge University Press, 2009.

\bibitem[Vieider(2024)]{vieider2024decisions}
Ferdinand~M. Vieider.
\newblock Decisions under uncertainty as {B}ayesian inference on choice
  options.
\newblock \emph{Management Science}, 2024.

\bibitem[Wilcox(2011)]{Wilcox}
Nathaniel~T. Wilcox.
\newblock {‘Stochastically more risk averse:’} a contextual theory of
  stochastic discrete choice under risk.
\newblock \emph{Journal of Econometrics}, 162\penalty0 (1):\penalty0 89--104,
  2011.

\bibitem[Wilcox(2021)]{wilcox2021utility}
Nathaniel~T. Wilcox.
\newblock Utility measurement: Some contemporary concerns.
\newblock In \emph{A Modern Guide to Philosophy of Economics}, pages 14--27.
  Edward Elgar Publishing, 2021.

\end{thebibliography}
\end{document}